\documentclass[11pt]{article}
\usepackage{graphicx} 
\usepackage{float}
\usepackage{amsmath}
\usepackage{amssymb}
\usepackage{amsfonts}
\usepackage{multirow}
\usepackage{amsthm}
\usepackage{color}
\usepackage{tikz}

\usepackage{algorithm}
\usepackage{algpseudocode}

\usepackage{changepage}
\usepackage{bbm}

\newtheorem{theorem}{Theorem}
\newtheorem{lemma}{Lemma}
\usepackage{natbib}

\newtheorem{proposition}{Proposition}

\usepackage{textcomp,marvosym}
\usepackage[margin=1in]{geometry}

\newtheorem{rem}{Remark}[section]
\newtheorem{defn}[theorem]{Definition}

\newcommand \bfmu{\boldsymbol \mu}
\newcommand \bfS{\boldsymbol \Sigma}

\newcommand{\norm}[1]{\left\Vert#1\right\Vert}

\makeatletter
\renewcommand\@biblabel[1]{}
\makeatother

\newcommand{\bfX}{{\bf X}}
\newcommand{\bfx}{{\bf x}}

\usepackage{hyperref}
\usepackage{nameref}
\title{Powerful Large Scale Inference in High Dimensional Mediation Analysis}
\author{Asmita Roy$^1$, Xianyang Zhang$^2$\footnote{Corresponding author: zhangxiany@stat.tamu.edu}\smallskip \\
	$^1$Department of Biostatistics, Johns Hopkins University\\
    $^2$Department of Statistics, Texas A\&M University}
\date{}
\begin{document}
\maketitle
\begin{abstract}
In genome-wide epigenetic studies, determining how exposures (e.g., Single Nucleotide Polymorphisms) affect outcomes (e.g., gene expression) through intermediate variables, such as DNA methylation, is a key challenge. Mediation analysis provides a framework to identify these causal pathways; however, testing for mediation effects involves a complex composite null hypothesis. Existing methods, such as Sobel's test or the Max-P test, are often underpowered in this context because they rely on null distributions determined under only a subset of the null space and are not optimized for the multiple testing burden inherent in high-dimensional data. To address these limitations, we introduce MLFDR (Mediation Analysis using Local False Discovery Rates), a novel method for high-dimensional mediation analysis. MLFDR leverages local false discovery rates, calculated from the coefficients of structural equation models, to construct an optimal rejection region. We demonstrate theoretically and through simulation that MLFDR asymptotically controls the false discovery rate and achieves superior statistical power compared to recent high-dimensional mediation methods. In real data applications, MLFDR identified 5\%–50\% more significant mediators than existing methods, demonstrating its ability to uncover biological signals missed by conventional approaches.
\end{abstract}

\section{Introduction}
Mediation analysis serves as a critical tool for deciphering the biological mechanisms underlying genetic associations with diseases identified in Genome-Wide Association Studies (GWAS). By bridging the gap between genetic variants and clinical outcomes, mediation analysis reveals intermediate pathways and elucidates causal relationships. As GWAS continues to uncover a vast number of genetic associations, translating these findings into actionable insights for precision medicine and therapeutic development becomes increasingly important. For instance, cigarette smoking is known to alter DNA methylation and gene expression \citep{maas2020smoking}; concurrently, DNA methylation often regulates gene expression directly \citep{dhar2021dna,moore2013dna}. Investigating the mediating effect of DNA methylation on gene expression—particularly in the presence of environmental exposures like smoking—is therefore essential. However, these analyses are complicated by high-dimensional outcomes and clinical confounders, such as patient age, which influences both gene expression and DNA methylation heterogeneity \citep{harris2017age,somel2006gene}. This article addresses the statistical challenges inherent in such high-dimensional mediation problems.

Historically, \citet{baron1986moderator} introduced the regression-based definition of mediation analysis, often referred to as the ``product of coefficients method,'' which examines the significance of the product of the exposure-mediator and mediator-outcome coefficients. More recently, the literature has expanded through the ``counterfactual framework'' \citep{robins1992identifiability, pearl2012causal, vanderweele2010odds, valeri2013mediation, vanderweele2016mediation, tchetgen2011causal, lange2011direct}, which provides a causal interpretation for natural direct and indirect effects across various models, including those with non-linearities and binary or survival outcomes.

Let $X$ denote the exposure, $M_i$ the $i$th mediator, and $Y$ the outcome. Under the product of coefficients approach, mediation analysis tests the null hypothesis $H_{0,i}: \alpha_i\beta_i = 0$, where $\alpha_i$ represents the effect of $X$ on $M_i$, and $\beta_i$ represents the effect of $M_i$ on $Y$. This creates a composite null hypothesis comprising three distinct cases: (i) $\alpha_i = 0, \beta_i \neq 0$; (ii) $\alpha_i \neq 0, \beta_i = 0$; or (iii) $\alpha_i = 0, \beta_i = 0$. Assuming no unmeasured confounders, classical tests like Max-P \citep{mackinnon2002comparison} and Sobel’s test \citep{sobel1982asymptotic} are known to be conservative under case (iii), as statistical inference is typically derived from distributions determined by cases (i) and (ii). In genome-wide studies, however, the sparse nature of omics data implies that $\alpha_i = 0$ and $\beta_i = 0$ hold for the majority of markers. Recent methods such as JS-mixture (HDMT) \citep{Dai2022} and DACT \citep{Liu2022} attempt to address this by explicitly modeling the composite nature of the null. JS-mixture improves power by using a mixture-null distribution of maximum p-values, adapting \citet{storey2002direct}'s procedure to estimate component proportions. DACT estimates the proportions of null $\alpha_i$ and $\beta_i$ separately to combine case-specific p-values. However, \citet{yang2025causal} recently demonstrated that DACT suffers from False Discovery Rate (FDR) inflation under dense alternatives and proposed a modified version (MDACT) that computes the statistic's distribution via numerical integration to improve p-value accuracy.

While JS-mixture and MDACT offer improvements over classical methods, they are not theoretically optimal regarding power. The FDR literature is broadly divided into p-value-based and local FDR-based rejection regions. Local FDR, a Bayesian approach, ranks hypotheses by the posterior probability that a case is null given the observed statistics; this ranking often differs from that based on p-values. \citet{Sun2007} demonstrated that, except in cases of symmetric alternatives, local FDR and p-value-based orderings diverge. Furthermore, \citet{Sun2007} proved that the local FDR-based oracle procedure is optimal: among all methods controlling the marginal FDR (mFDR), the local FDR approach yields the highest number of rejections. While the power advantage is negligible for symmetric alternatives, it becomes significant when the alternative distribution is asymmetric. Motivated by these theoretical properties, we propose MLFDR, a local FDR-based screening algorithm designed specifically for high-dimensional mediation analysis. Our contributions to the literature are as follows:

\begin{enumerate}
    \item We extend the concept of local FDR to the composite null hypothesis setting, deriving a screening rule with a closed-form expression for the corresponding false discovery proportion (FDP).
    \item We validate the method across a diverse array of data types—including continuous and binary variables, with and without confounders or exposure-mediator interactions—demonstrating robust performance across various model specifications. We also illustrate the method's efficacy in multiple mediator setups with univariate, clinical, or survival outcomes.
    \item MLFDR offers optimal power improvement over existing methods while maintaining asymptotic FDR control. Extensive simulations confirm its superiority over MDACT and HDMT in terms of power and error rate control.
    \item We provide theoretical guarantees for the identifiability and global optimality of our model under relatively mild assumptions, proving FDR control for both the oracle and adaptive procedures.
\end{enumerate}

The remainder of this paper is organized as follows. Section \ref{methodoverview} outlines the screening procedure for detecting significant mediators. Section \ref{sim} presents simulation studies. Section \ref{sec:extension} discusses extensions of MLFDR to composite alternatives and latent factor models, which can account for unmeasured confounding and pleiotropy. Section \ref{realdata} provides an in-depth analysis of Prostate Cancer data and Lung cancer data from The Cancer Genome Atlas (TCGA), exploring SNP-CpG-gene expression pathways and causal pathways between smoking habits and gene expression, respectively. Section \ref{Methods} details the methodology. An R package implementing the method is available at \url{https://github.com/asmita112358/MLFDR} as well as in CRAN. Theorems proving the large-sample FDR control of MLFDR are provided in the Supplementary Materials.

Finally, we distinguish our approach from other recent efforts in high-dimensional mediation analysis. \citet{zhang2021mediation}, \citet{yu2021high}, and \citet{perera2022hima2} address the multiple mediator problem specifically for survival outcomes. Closer to our framework, \citet{sun2023testing} and \citet{ding2024amdp} utilize local FDR-based rejection regions; the former approximates the alternative as a mixture of Gaussian distributions, while the latter constructs regions based on p-values. Our work advances this domain in two specific aspects: (i) we incorporate a general prior for the coefficients $\alpha$ and $\beta$ to estimate the \textit{exact} posterior density for computing the local FDR, rather than relying on approximations; and (ii) we offer theoretical guarantees for the local FDR estimates obtained via the EM algorithm, a property not previously established in this context.

\section{Method Overview}\label{methodoverview}
This section outlines the workflow of MLFDR; a schematic representation of the framework is provided in Figure \ref{fig:schematic}. Consider a study involving $n$ independent samples. For each testing unit $i = 1, \dots, m$, we observe an exposure variable $X_i$, a mediator $M_i$, and an outcome $Y_i$. Biologically, these variables may represent distinct contexts: for example, $X$ may denote a patient's smoking history (shared across $i$), with $\{M_i\}_{i=1}^m$ representing CpG methylation sites and $\{Y_i\}_{i=1}^m$ representing gene expression levels. Alternatively, the analysis may focus on the functional impact of Single Nucleotide Polymorphisms ($X_i$) on gene expression ($Y_i$) as mediated by CpG methylation ($M_i$) \citep{Dai2022}. 

The mediation model posits that the exposure $X_i$ influences the outcome $Y_i$ through the intermediate variable $M_i$, rather than solely through a direct relationship. We denote the coefficient for the exposure-mediator relationship ($X_i \to M_i$) as $\alpha_i$, and the coefficient for the mediator-outcome relationship ($M_i \to Y_i$) as $\beta_i$. In Figure \ref{fig:schematic}, solid arrows indicate these direct effects. 

We aim to test the composite null hypothesis against the alternative for each unit $i$:
\begin{equation}
    H_{0,i}: \alpha_i\beta_i = 0 \quad \text{versus} \quad H_{1,i}: \alpha_i\beta_i \neq 0, \quad i=1,2,\dots,m.
\end{equation}
The composite null hypothesis $H_{0,i}$ can be decomposed into three disjoint component nulls, $H_{0,i}=H_{00,i}\cup H_{01,i}\cup H_{10,i}$, defined as:
\begin{equation}
    \begin{aligned}
     H_{00,i} &: \alpha_i = 0 \text{ and } \beta_i = 0, \\
     H_{10,i} &: \alpha_i \neq 0 \text{ and } \beta_i = 0, \\
     H_{01,i} &: \alpha_i = 0 \text{ and } \beta_i \neq 0,
    \end{aligned}
\end{equation}
for $i = 1,2,\dots,m.$


\begin{figure}[ht]
    \centering
    \resizebox{16cm}{!}{
   \input{fig2}
   }
    \caption{Schematic diagram of MLFDR.}
    \label{fig:schematic}
\end{figure}

We consider a mixture prior for $(\alpha_i, \beta_i)$, where the probability of each disjoint component nulls $H_{00}, H_{10}$ and $H_{01}$ occur with probability $\pi_{00}, \pi_{10}$ and $\pi_{01}$ respectively. The marginal prior distributions of $\alpha_i$ and $ \beta_i$, respectively, are degenerate zero under the null and follow a normal prior with an unknown mean and variance under the alternative. The marginal distribution of the least squares coefficient estimates $\{\hat{\alpha}_i, \hat{\beta}_i\}$ given the latent states is computed, and the unknown parameters including the null proportions are estimated using EM algorithm.

 Using these estimates, we compute the local false discovery rate (local FDR) for each coefficient pair, denoted as $\widehat{\text{lfdr}}_{i}$ for $i = 1, \dots, m$. As the local FDR represents the posterior probability that the $i$-th hypothesis is null given the observed statistics, a lower value indicates stronger evidence against the null. Consequently, we define the rejection region for the composite null hypothesis as $\widehat{\text{lfdr}}_i \leq \delta$. The threshold $\delta$ is determined adaptively using the step-up procedure proposed by \citet{Sun2007}. The complete algorithm is detailed in the Methods section.

Additionally, we introduce an extended algorithm (MLFDR2) which can deal with scenarios where the marginal priors of  $\alpha_i$ and $\beta_i$ follow mixture normal distributions under the alternative. A composite alternative leads to a joint distribution of $\{\hat{\alpha}_i, \hat{\beta}_i\}$ with more than 4 mixture components, which can often be computationally burdensome. We introduce a two-step EM alogrithm that estimates the parameters of the marginal distributions of $\{\hat{\alpha}_i\}$ and $\{\hat{\beta}_i\}$ in the first step, then uses these estimates to run another EM algorithm that computes the probabilities of each mixture component.

 We also discuss another extension using Surrogate Variable Analysis (\cite{leek2007capturing}) which can account for unmeasured confounders and pleiotropy in the model. Details are presented in Section \ref{sec:extension}.

\section{Simulation Studies}\label{sim}
We evaluate the performance of MLFDR through extensive simulations under two distinct mixture proportion scenarios: a \textit{dense} alternative and a \textit{sparse} alternative. Following the setups in \citet{Dai2022}, the latent class probabilities are defined as:
\begin{itemize}
    \item \textbf{Dense alternative:} $(\pi_{00}, \pi_{10}, \pi_{01}, \pi_{11}) = (0.4, 0.2, 0.2, 0.2).$
    \item \textbf{Sparse alternative:} $(\pi_{00}, \pi_{10}, \pi_{01}, \pi_{11}) = (0.88, 0.05, 0.05, 0.02).$
\end{itemize}

We consider sample sizes of $n \in \{100, 300\}$ and fix the number of mediators at $m = 1000$. The parameter controlling the signal strength of mediation, $\tau$, varies from $0.1$ to $1.9$ in increments of $0.2$. The non-zero coefficients are generated as $\alpha_i = 0.05\tau + h_i$ (under $H_{10}$ and $H_{11}$) and $\beta_i = -0.5\tau + g_i$ (under $H_{01}$ and $H_{11}$), where the noise terms follow $h_i \sim N(0, 1/n)$ and $g_i \sim N(0, 4/n)$.

We compare the empirical FDR and power of MLFDR against two competing methods: MDACT \citep{yang2025causal} and HDMT (JS-mixture) \citep{Dai2022}. The simulation settings are detailed below.

\begin{enumerate}
    \item \textbf{Linear Model.} The exposure is univariate with $X \sim \text{Ber}(0.1)$.
    \begin{equation}
    \begin{aligned}
    M_i &= X\alpha_i + e_i, \\
    Y_i &= M_i\beta_i + X\gamma_i + \epsilon_i,
    \end{aligned}
    \end{equation}
    where $\gamma_i \sim N(1, 0.5)$, and error terms $e_i, \epsilon_i \stackrel{i.i.d.}{\sim} N(0,1)$. The results for this setting are summarized in Figure \ref{lm}.

    \item \textbf{Linear Model with measured Confounder.} 
    The exposure is univariate with $X \sim \text{Ber}(0.1)$. We introduce a confounder $Z \sim N(0,1)$:
    \begin{equation}
    \begin{aligned}
    M_i &= X\alpha_i + \theta_iZ + e_i,\\
    Y_i &= M_i\beta_i + X\gamma_i + \delta_iZ + \epsilon_i,
    \end{aligned}
    \end{equation}
    where the confounder effects are drawn independently from $\theta_i, \delta_i \sim U(0, 0.5)$. The results are presented in Figure \ref{lm_Z}.

    \item \textbf{Binary Outcome.} The exposure is univariate with $X \sim \text{Ber}(0.1)$. The outcome $Y_i$ is binary:
    \begin{equation}
    \begin{aligned}
    M_i &= X\alpha_i + e_i,\\
    \text{logit}\{\mathbb{P}(Y_i=1)\} &= M_i\beta_i + X\gamma_i.
    \end{aligned}
    \end{equation}
    The results are displayed in Figure \ref{binY}.
\end{enumerate}
 
Across all settings, the three methods demonstrated satisfactory FDR control. However, MLFDR consistently exhibited the highest power. Specifically, MLFDR achieved an average power improvement of $10.83\%$ over MDACT and $12.23\%$ over HDMT under dense alternatives. Under sparse alternatives, MLFDR maintained its advantage with an average improvement of $7.47\%$ over MDACT and $8.51\%$ over HDMT.

\begin{figure}[ht]
    \centering
    \includegraphics[width=\linewidth]{simulation_lm.png}
    \caption{FDR and power comparison for the linear model (Setting 1). Results are displayed for both sparse and dense alternatives. Gray ribbons indicate error margins.}
    \label{lm}
\end{figure}

\begin{figure}[ht]
    \centering
    \includegraphics[width=\linewidth]{simulation_Z.png}
    \caption{FDR and power comparison for the linear model with measured confounders (Setting 2). Results are displayed for both sparse and dense alternatives. Gray ribbons indicate error margins.}
    \label{lm_Z}
\end{figure}

\begin{figure}[ht]
    \centering
    \includegraphics[width=\linewidth]{simulation_binY.png}
    \caption{FDR and power comparison for the linear model with binary outcomes (Setting 3). Results are displayed for both sparse and dense alternatives. Gray ribbons indicate error margins.}
    \label{binY}
\end{figure}

\begin{figure}[ht]
    \centering
    \includegraphics[width=\linewidth]{simulation_ME.png}
    \caption{FDR and power comparison for the linear model with unknown mediator-exposure interactions. Results are displayed for both sparse and dense alternatives. Gray ribbons indicate error margins.}
    \label{ME}
\end{figure}

\begin{figure}[ht]
    \centering
    \includegraphics[width=\linewidth]{simulation_pl.png}
    \caption{FDR and power comparison for the linear model with unmeasured confounders. Results are displayed for both sparse and dense alternatives. Gray ribbons indicate error margins.}
    \label{pl}
\end{figure}

\section{Extensions} \label{sec:extension}
\subsection{Composite Alternatives}
In this setting, the coefficients $(\alpha, \beta)$ follow a Gaussian mixture distribution. The posterior distribution of the estimated coefficients is given by:
\begin{align*}
    \sqrt{n}\hat{\alpha}_i &\sim p_0 N(0, \sigma_{i1}^2) + p_1N(\mu_1, \sigma_{i1}^2 + \kappa_1) + p_2N(\mu_2, \sigma_{i1}^2 + \kappa_2),\\
    \sqrt{n}\hat{\beta}_i &\sim q_0 N(0, \sigma_{i2}^2) + q_1N(\theta_1, \sigma_{i2}^2 + \psi_1) + q_2N(\theta_2, \sigma_{i2}^2 + \psi_2).
\end{align*}
Under this framework, the null hypothesis corresponds to a mixture of 5 bivariate Gaussian distributions, while the alternative hypothesis comprises a mixture of 4 bivariate Gaussian distributions. The parameters are estimated using a two-step EM algorithm, the details of which are provided in the Supplementary Materials.

In our simulations, we set the mixture weights to $(p_0, p_1, p_2) = (0.54, 0.18, 0.28)$ and $(q_0, q_1, q_2) = (0.6, 0.05, 0.35)$. The variance parameters were fixed at $\kappa_1 = 1, \kappa_2 = 2, \psi_1 = 1.5, \psi_2 = 2$. The mean parameters were defined as functions of the mediation signal strength $\tau$: $\mu_1 = 0.05\tau$, $\mu_2 = -0.5\tau$, $\theta_1 = 0.9\tau$, and $\theta_2 = -0.01\tau$, where $\tau$ ranges from $0.1$ to $1.9$ in increments of $0.2$.

Figure \ref{composite_alt} presents the results for $n \in \{100, 300\}$ with $m = 1000$. All three methods maintained satisfactory FDR control, with MLFDR demonstrating the highest power.

\begin{figure}[ht]
    \centering
    \includegraphics[width=\linewidth]{simulation_comp_alt.png}
    \caption{FDR and power comparison for composite alternatives. Results are displayed for varying degrees of mediation ($\tau$). Gray ribbons indicate error margins.}
    \label{composite_alt}
\end{figure}

\subsection{Latent factors}
Unmeasured latent factors may be addressed by surrogate variable analysis (\cite{leek2007capturing}). Briefly, surrogate variable analysis considers the following model:
\begin{equation}
M_i = \mu_i + \alpha_i X + \phi_i Z + \sum_{l = 1}^L \gamma^{(i)}_l g_l+e_i,
\end{equation}
where $X,Z$ are measured covariates, and $g_1, g_2, \ldots g_L$ are unmeasured latent factors. Surrogate variable analysis produces a set of $K$ mutually orthogonal vectors $\hat{u}_1, \hat{u}_2, \ldots, \hat{u}_K$ (where $K \le L$), which span the same linear space as the latent factors. Thus, the original equation may be re-written as:
\begin{align*}
    M_i &= \mu_i + \alpha_i X + \phi_i Z + \sum_{k = 1}^K \lambda_k^{(i)}\hat{u}_k + e_i.
\end{align*}
These estimated factors, collected into a matrix $\hat{\mathbf{U}}_M$, account for unmeasured confounding in the exposure-mediator relationship.

For the mediator-outcome relationship, the latent factors may be modeled as follows:
\begin{align}\label{sva1}
    Y_i = \nu_i + \beta_i M_i + \gamma_i X+\delta_iZ + \sum_{l = 1}^L \eta^{(i)}_l g_l + \epsilon_i.
\end{align}
In this setting, the latent terms may account for: 1) unmeasured confounding; 2) measured confounders with unknown relationships to the outcome (e.g., global batch effects); and 3) pleiotropy, where $Y_i$ is influenced by mediators other than $M_i$. 

Direct application of SVA to model (\ref{sva1}) would require estimating surrogate variables for each mediator-outcome pair iteratively, leading to a computational bottleneck. To address this, we propose a global factor adjustment. We estimate a second set of surrogate variables, $\hat{\mathbf{U}}_Y$, based on the outcome null model (excluding mediators):
\begin{align}\label{sva2}
    Y_i = \nu_i + \gamma_i X + \delta_i Z + \sum_{l = 1}^L \eta^{(i)}_l g_l + \epsilon_i.
\end{align}
We then use the combined set of latent factors $\hat{\mathbf{U}}_M$ (derived from mediators) and $\hat{\mathbf{U}}_Y$ (derived from outcome residuals) to model the mediator-outcome relationship. The validity of using $\hat{\mathbf{U}}_Y$ from (\ref{sva2}) relies on the assumption that any single mediator $M_i$ contributes a relatively small amount of variance to the global outcome matrix $\textbf{Y}$. 

The full procedure is summarized in Algorithm \ref{alg:global_adjustment}. We implemented surrogate variable analysis via the Bioconductor package \texttt{sva} (\cite{leek2012sva}).

We apply this method to two data generating scenarios to demonstrate its performance.

\noindent
\textbf{Unknown Mediator-Exposure Interactions.}   
\begin{equation}
    \begin{aligned}
    M_i &= \alpha_iX +\delta_i Z + e_i,\\
    Y_i &= M_i\beta_i + X\gamma_i + \zeta_iZ + X\sum_{j \in S}\theta_jM_j + \epsilon_i,
    \end{aligned}
    \end{equation}
    where $S$ is a randomly selected subset of indices $\{1,\ldots, m\}$ with $|S| = 20$, treated as unknown during model fitting. The results are shown in Figure \ref{ME}.
    
\noindent
\textbf{Unmeasured Confounding and Pleiotropy.} 
    \begin{equation}
        \begin{aligned}
            M_i &= X\alpha_i + \theta_i Z + 0.4 Z_1 + 0.5Z_2 + e_i,\\
            Y_i &= M_i\beta_i + X\gamma_i + \delta_iZ - 0.5Z_1 + \sum_{j \in S}M_j\kappa_j + \epsilon_i,
        \end{aligned}
    \end{equation}
where $Z_1$ and $Z_2$ represent unmeasured confounders not included in the model fitting. $S$ is a randomly selected subset of indices $\{1,\ldots, m\}$ with $|S| = 20$. The term $\sum_{j \in S}M_j\kappa_j$ represents dense pleiotropy (the effect of other mediators on $Y_i$), which acts as an additional source of unmeasured variation. The results are shown in Figure \ref{pl}.

All methods appear to have satisfactory FDR control. MLFDR is uniformly more powerful than HDMT and MDACT in all cases, with better margin of improvement in dense alternatives.

\section{Real Data Analysis}\label{realdata}
\subsection{TCGA Prostate Cancer Data}
We apply MLFDR to the Prostate Cancer dataset from The Cancer Genome Atlas (TCGA), previously analyzed by \citet{Dai2022}. The study involves mediation analysis for 147 prostate cancer risk SNPs, integrated with DNA methylation and gene expression data from 495 samples. For each risk SNP, we identified CpG methylation probes within a 500 kb window and recorded the gene expression levels for the corresponding probes. This resulted in $m = 69,602$ SNP-CpG-Gene triplets for mediation testing.

In the first stage, we regressed CpG methylation on the SNPs, adjusting for the top 3 principal components (PCs) of genotypes, the top 15 PCs of CpG methylation, age at diagnosis, and pathological stage. From this, we obtained the slope estimates, variances, and p-values for the SNPs. In the second stage, gene expression was regressed on CpG methylation, conditional on the same set of covariates.

The estimated null proportion components $(\pi_{00}, \pi_{10}, \pi_{01})$ were $(0.51, 0.033, 0.41)$ for HDMT, compared to $(0.39, 0.004, 0.59)$ obtained via the EM algorithm in MLFDR.

Due to the wide spread of the methylation coefficients ($\beta$), we fitted a composite alternative Gaussian mixture model. The number of components, $d_2 = 8$, was selected based on the Akaike Information Criterion (AIC) (Figure \ref{TCGA_AIC}). Conversely, the SNP coefficients ($\alpha$) exhibited a narrower range ($-0.2$ to $0.4$) and were adequately modeled using a $d_1 = 2$ component Gaussian mixture.

\begin{figure}[H]
    \centering
    \includegraphics[width=0.7\linewidth]{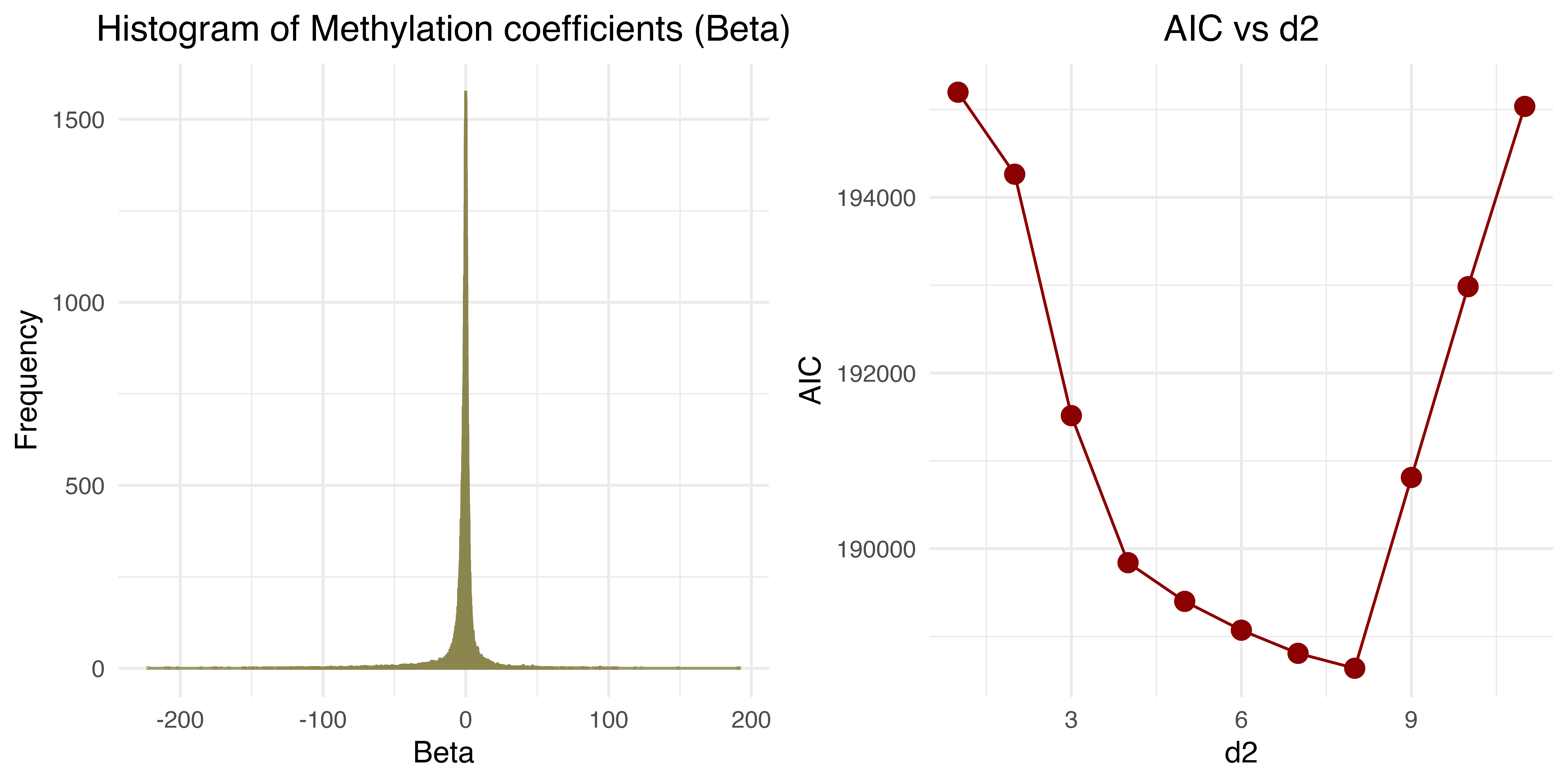}
    \caption{AIC for methylation coefficients when a $d_2 + 1$ component Gaussian mixture model was fit.}
    \label{TCGA_AIC}
\end{figure}

At an FDR threshold of $0.01$, HDMT identified 137 triplets, MLFDR identified 187, and MDACT identified 180. Figure \ref{venn} displays a Venn diagram of the overlapping discoveries, along with the number of rejections across FDR cutoffs ranging from $0.001$ to $0.05$. MLFDR consistently detected more pathways than the competing methods.

\begin{figure}[H]
    \centering
    \includegraphics[width=0.8\linewidth]{TCGA_Venn_rejplot.png}
    \caption{SNP-CpG-Gene triplets identified by MDACT, MLFDR, and HDMT out of 69,602 tests. The Venn diagram (left) corresponds to an FDR cutoff of 0.01.}
    \label{venn}
\end{figure}

Table \ref{tcga_tab} lists 10 additional pathways identified by MLFDR (but missed by HDMT or MDACT), ranked by local FDR. Notably, six of these triplets involve rs12653946, a known prostate cancer risk variant that influences the expression of \textit{IRX4}, a tumor suppressor gene in the prostate \citep{ha2012irx4}. Another pathway involves rs7767188, a risk SNP associated with prostate cancer through the expression of \textit{TRIM26} \citep{wu2020integrative}.

\begin{table}[ht]
\centering
\scalebox{0.9}{
\begin{tabular}{llllrr}
\hline
SNP & CpG Probe & Annotated Gene & Target Gene & $p_{\max}$ & Local FDR \\ 
\hline
rs7767188 & cg02749752 & TRIM26 & TRIM26 & 0.0017 & 0.0034 \\ 
rs12653946 & cg12830271 & -- & IRX4 & 0.0019 & 0.0058 \\ 
rs3096702 & cg19609334 & TNXB & TNXB & 0.0019 & 0.0060 \\ 
rs3129859 & cg03520342 & HLA-DMA & HLA-DMA & 0.0019 & 0.0125 \\ 
rs12653946 & cg00085370 & -- & IRX4 & 0.0022 & 0.0067 \\ 
rs12653946 & cg07144328 & -- & IRX4 & 0.0022 & 0.0063 \\ 
rs12653946 & cg07278634 & -- & IRX4 & 0.0028 & 0.0077 \\ 
rs12653946 & cg06446548 & -- & NDUFS6 & 0.0028 & 0.0288 \\ 
rs12653946 & cg03225093 & -- & IRX4 & 0.0029 & 0.0083 \\ 
rs5945619 & cg10581449 & NUDT11 & NUDT11 & 0.0030 & 0.0077 \\ 
\hline
\end{tabular}
}
\caption{Top 10 additional pathways detected by MLFDR in TCGA Prostate Cancer Data (ranked by $p_{\max}$). These pathways were not detected by HDMT or MDACT.}
\label{tcga_tab}
\end{table}

These empirical results align with our simulation studies: HDMT yielded the fewest discoveries, followed by MDACT, with MLFDR providing the highest detection rate. We note that the improvement of MLFDR over MDACT is modest in this application. This is likely attributable to the symmetric distribution of $\alpha$ and $\beta$; as noted by \citet{Sun2007}, local FDR-based tests offer limited power gains over p-value-based methods when the alternative distribution is symmetric around the null.

\subsection{TCGA Lung Squamous Cell Carcinoma}

We further extend our analysis to the TCGA Lung Squamous Cell Carcinoma (LUSC) dataset to investigate the mediating role of CpG methylation in the relationship between smoking history and gene expression. The data were acquired using the R package \texttt{UCSCXenaTools} \citep{wang2019ucscxenatools}. We restricted the analysis to primary tumor samples, resulting in a sample size of $n = 379$ after preprocessing.

Smoking history was quantified by pack-years. Using the publicly available probe map data from the TCGA website, each CpG probe was mapped to the expression profiles of potentially multiple genes; each unique CpG-gene pair was treated as a distinct candidate mediation pathway.
The analysis proceeded in two stages. In the first stage, CpG methylation beta values were regressed on smoking history (pack-years). In the second stage, a multiple linear regression was fitted for each gene, including all CpG probes mapped to that gene as predictors. The coefficients and p-values for each specific CpG-gene pair were then extracted to form the mediation hypotheses. All models were adjusted for potential confounders, including sex and the age at initial diagnosis. In total, 319,761 CpG-gene pathways were evaluated.

At an FDR threshold of 0.01, HDMT identified 13 pathways, MDACT identified 25 pathways, and MLFDR identified 44 pathways (Figure \ref{LUSC}). The results highlight the increased power of MLFDR in detecting subtle mediation signals.
Table \ref{LUSCtab} details the top 10 additional pathways detected by MLFDR (ranked by local FDR) that were not identified by the competing methods. Several of these findings are supported by existing literature. For instance, \citet{wang2013wdr66} discuss the relevance of \textit{WDR66} in lung cancer progression, while \citet{chen2025construction} highlights the role of \textit{LY6K} as a potential therapeutic target in lung squamous cell carcinoma. \cite{yang2018t} links \textit{TCIRG1} to metastatic potential of hepatocellular carcinoma.

\begin{figure}[H]
    \centering
    \includegraphics[width=0.8\linewidth]{TCGA_LUSC_Venn_rejplot.png}
    \caption{Smoking-CpG-Gene pathways detected by MDACT, MLFDR, and HDMT out of 319,761 tests. The Venn diagram (left) displays overlaps at an FDR cutoff of 0.01, while the plot (right) shows detection counts across varying FDR thresholds.}
    \label{LUSC}
\end{figure}

\begin{table}[ht]
\centering
\begin{tabular}{rllrr}
  \hline
  Rank & Gene & CpG Probe & $p_{\max}$ & Local FDR \\ 
  \hline
  1 & HOXB6 & cg20591728 & $6.6 \times 10^{-6}$ & 0.0039 \\ 
  2 & MYLIP & cg04641165 & $1.7 \times 10^{-5}$ & 0.0040 \\ 
  3 & B3GALT2 & cg16712103 & $1.8 \times 10^{-5}$ & 0.0100 \\ 
  4 & ZNF287 & cg16964464 & $2.2 \times 10^{-5}$ & 0.0085 \\ 
  5 & TCIRG1 & cg20484322 & $3.6 \times 10^{-5}$ & 0.0131 \\ 
  6 & WDR66 & cg03560652 & $3.8 \times 10^{-5}$ & 0.0145 \\ 
  7 & ENO3 & cg07333510 & $4.5 \times 10^{-5}$ & 0.0139 \\ 
  8 & ADORA2B & cg21501163 & $5.8 \times 10^{-5}$ & 0.0197 \\ 
  9 & GNA14 & cg06617692 & $6.4 \times 10^{-5}$ & 0.0231 \\ 
  10 & LY6K & cg16809304 & $8.1 \times 10^{-5}$ & 0.0148 \\ 
  \hline
\end{tabular}
\caption{Top 10 additional pathways detected by MLFDR in the TCGA LUSC dataset (ranked by $p_{\max}$). These pathways were not detected by HDMT or MDACT.}
\label{LUSCtab}
\end{table}
To evaluate the FDR of the three methods, we selected a subset of the data with $m = 1,000$ tests, and permuted the samples to create a ``global null" scenario. For each permutation, the smoking-Cpg and CpG-gene expression model is fit, and the three methods are implemented. At a given permutation, if any method has non-zero rejections, the False Positive rate for that permutation is recorded as 1 for that method. This procedure is repeated over 100 permutations, and the number of rejections is recorded for FDR levels varying from 0.0001 to 0.2. Figure \ref{fig:globalnull} presents the proportion of cases in which each method came up with significant rejections under the global null. HDMT and MLFDR appear to have satisfactory control of the FDR, while MDACT appears to have some FDR inflation at low FDR levels.
\begin{figure}
    \centering
    \includegraphics[width=0.5\linewidth]{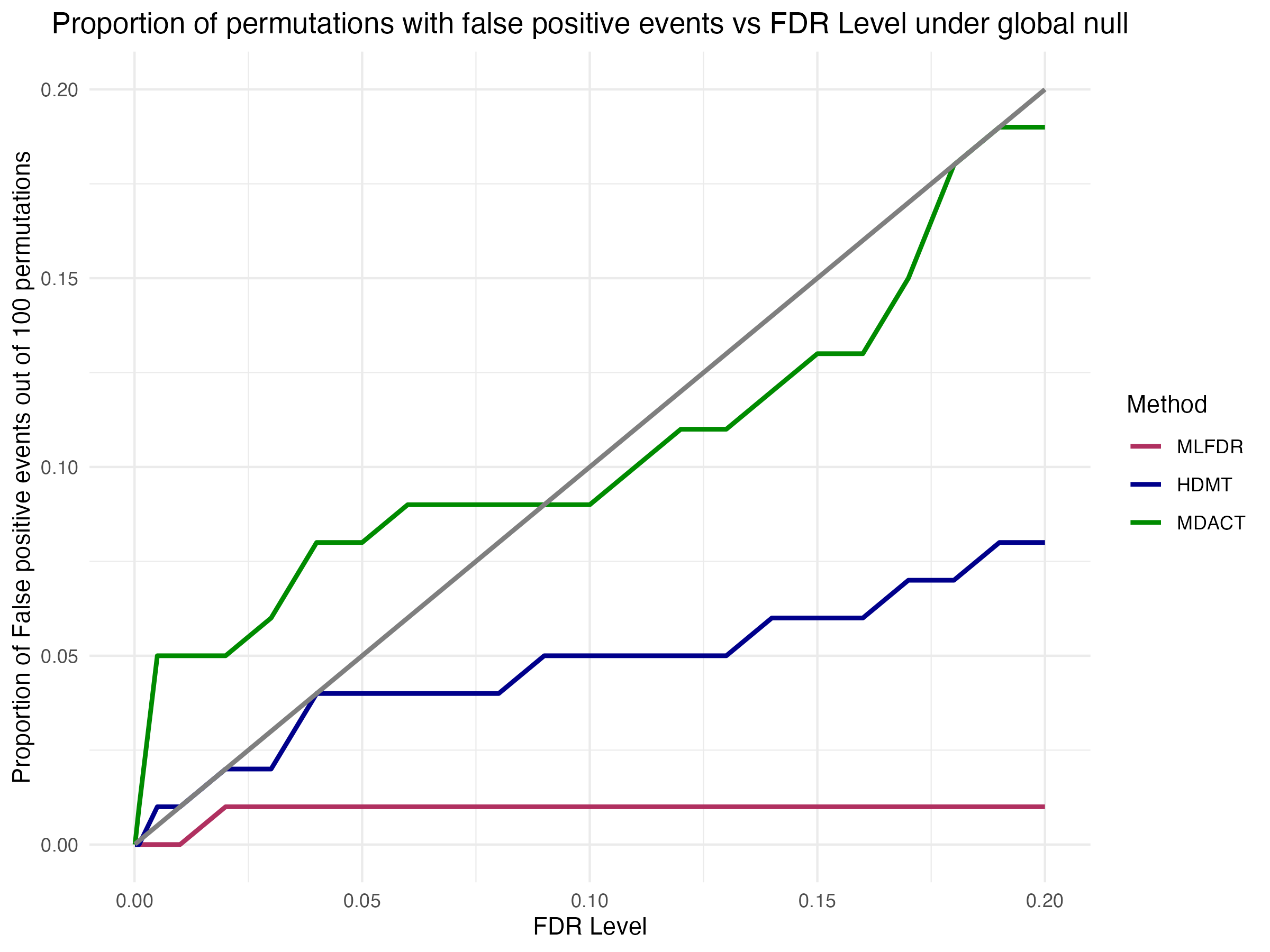}
    \caption{Smoking-CpG-gene pathways detected by all methods under the global null out of $m = 1,000$ tests. The proportion of cases reporting positive detection counts is reported across varying FDR thresholds. The gray solid line indicates target FDR levels.}
    \label{fig:globalnull}
\end{figure}

\section{Methods}\label{Methods}
\subsection{The Mediation Model}

Consider a univariate exposure $X$, a set of mediators $\{M_i\}_{i = 1}^m$, and $m$ outcomes $\{Y_i\}_{i = 1}^m$. We assume the following structural equation model:
\begin{equation}\label{sem}
    \begin{aligned}
    M_i &= X\alpha_i + e_i,\\
    Y_i &= M_i\beta_i + X\gamma_i +\epsilon_i,
    \end{aligned}
\end{equation}
where the error terms satisfy $(e_i,\epsilon_i)\perp (X,M_i)$ and are distributed as:
\begin{align*}
(e_i,\epsilon_i)\sim N\left(
\begin{pmatrix}
  0 \\
  0
\end{pmatrix},
\begin{pmatrix}
  \sigma_{i,a}^2 &0 \\
  0 & \sigma_{i,b}^2
\end{pmatrix}
\right).
\end{align*}
In settings with high-dimensional exposures, such as Single Nucleotide Polymorphisms (SNPs), we define the hypothesis based on the Exposure-Mediator-Outcome triplet $\{X_i, M_i, Y_i \}_{i = 1}^m$. In such cases, the common exposure $X$ in Equation (\ref{sem}) is replaced by the specific exposure $X_i$.

To describe the methodology, we focus on the formulation in Model (\ref{sem}). Given $n$ independent samples $\{(Y^j,X^j,M^j)\}^{n}_{j=1}$, our goal is to test the joint significance of $\alpha_i$ and $\beta_i$ for $1\leq i\leq m$.
Let $\mathbf{Y}_i=(Y_i^1,\dots,Y_i^n)'$, $\mathbf{X}=(X^1,\dots,X^n)'$, $\mathbf{M}_i=(M^1_i,\dots,M^n_i)'$, and define the projection matrix $\mathbf{P}=\mathbf{I}-\mathbf{X}(\mathbf{X}'\mathbf{X})^{-1}\mathbf{X}'$. The ordinary least squares (OLS) estimators are given by:
\begin{equation}\label{stat}
        \hat{\alpha}_{i} =(\mathbf{X}'\mathbf{X})^{-1}\mathbf{X}'\mathbf{M}_i, \quad 
        \hat{\beta}_{i} =(\mathbf{M}'_i\mathbf{P}\mathbf{M}_i)^{-1}\mathbf{M}'_i\mathbf{P}\mathbf{Y}_i.
\end{equation}
Conditional on $\mathbf{X}$ and $\mathbf{M}_i$, the estimators follow the distribution:
\begin{align*}
\sqrt{n}
\begin{pmatrix}
  \hat{\alpha}_{i}-\alpha_i \\
  \hat{\beta}_{i}-\beta_i
\end{pmatrix}
&=\sqrt{n}
\begin{pmatrix}
  (\mathbf{X}'\mathbf{X})^{-1}\mathbf{X}'\mathbf{e}_i \\
  (\mathbf{M}'_i\mathbf{P}\mathbf{M}_i)^{-1}\mathbf{M}'_i\mathbf{P}{\bf\epsilon}_i
\end{pmatrix} \\
&\overset{d}{=}
N\left(\begin{pmatrix}
         0 \\
         0
       \end{pmatrix},
\begin{pmatrix}
       \sigma_{i1}^2 & 0 \\
       0 & \sigma_{i2}^2
     \end{pmatrix}
     \right),
\end{align*}
where $\sigma_{i1}^2 = n\sigma_{i,a}^2(\mathbf{X}'\mathbf{X})^{-1}$ and $\sigma_{i2}^2= n\sigma_{i,b}^2 (\mathbf{M}_i'\mathbf{P}\mathbf{M}_i)^{-1}$. Note that $\hat{\alpha}_{i}$ and $\hat{\beta}_{i}$ are independent, as $\mathbb{E}[(\mathbf{X}'\mathbf{X})^{-1}\mathbf{X}'\mathbf{e}_i{\bf\epsilon}'_i\mathbf{P}\mathbf{M}_i(\mathbf{M}'_i\mathbf{P}\mathbf{M}_i)^{-1}]=0$. 
\begin{algorithm}[!h]
\caption{Two-Step Global Factor Adjustment for High-Dimensional Mediation}
\label{alg:global_adjustment}
\begin{algorithmic}[1]
\Require Outcome matrix $\mathbf{Y} \in \mathbb{R}^{n \times m}$, Mediator matrix $\mathbf{M} \in \mathbb{R}^{n \times m}$, Exposure vector $\mathbf{X} \in \mathbb{R}^{n \times 1}$, Covariates $\mathbf{Z} \in \mathbb{R}^{n \times q}$.
\Ensure Adjusted estimates for $\hat{\alpha}_i$ and $\hat{\beta}_i$ for $i=1,\dots,m$.

\Statex \textbf{Step 1: Estimate Latent Factors for Mediators ($\hat{\mathbf{U}}_M$)}
\State Regress $\mathbf{M}$ on $\mathbf{X}$ and $\mathbf{Z}$ to obtain the residual matrix $\hat{\mathbf{R}}_M$:
\[
\hat{\mathbf{R}}_M = \mathbf{M} - \left( \mathbf{X}\hat{\mathbf{A}} + \mathbf{Z}\hat{\mathbf{\Phi}} \right)
\]
\State Perform SVA on $\hat{\mathbf{R}}_M$ to extract the top $k_M$ principal components:
\[
\hat{\mathbf{U}}_M \leftarrow \text{SVA}(\hat{\mathbf{R}}_M, k_M)
\]

\Statex \textbf{Step 2: Estimate Latent Factors for Outcomes ($\hat{\mathbf{U}}_Y$)}
\State Regress $\mathbf{Y}$ on $\mathbf{X}$ and $\mathbf{Z}$ (excluding $\mathbf{M}$) to obtain the residual matrix $\hat{\mathbf{R}}_Y$:
\[
\hat{\mathbf{R}}_Y = \mathbf{Y} - \left( \mathbf{X}\hat{\mathbf{\Gamma}} + \mathbf{Z}\hat{\mathbf{\Delta}} \right)
\]
\State Perform SVA on $\hat{\mathbf{R}}_Y$ to extract the top $k_Y$ principal components:
\[
\hat{\mathbf{U}}_Y \leftarrow \text{SVA}(\hat{\mathbf{R}}_Y, k_Y)
\]

\Statex \textbf{Step 3: Mediation Analysis with Global Adjustment}
\For{$i = 1, \dots, m$}
    \State \textit{Mediator Model:} Regress $M_i$ on $\mathbf{X}$, $\mathbf{Z}$, and $\hat{\mathbf{U}}_M$:
    \[
    M_i = \mathbf{X}\alpha_i + \mathbf{Z}\phi_i + \hat{\mathbf{U}}_M \boldsymbol{\lambda}_{M,i} + e_i
    \]
    \State \textit{Outcome Model:} Regress $Y_i$ on $M_i$, $\mathbf{X}$, $\mathbf{Z}$, $\hat{\mathbf{U}}_M$, and $\hat{\mathbf{U}}_Y$:
    \[
    Y_i = M_i \beta_i + \mathbf{X}\gamma_i + \mathbf{Z}\delta_i + \hat{\mathbf{U}}_M \boldsymbol{\lambda}_{Y1,i} + \hat{\mathbf{U}}_Y \boldsymbol{\lambda}_{Y2,i} + \epsilon_i
    \]
    \State Extract coefficients $\hat{\alpha}_i, \hat{\beta}_i$ and their standard errors for MLFDR.
\EndFor
\end{algorithmic}
\end{algorithm}
We test the composite null hypothesis against the alternative:
\begin{equation}
    H_{0,i}: \alpha_i\beta_i = 0 \quad \text{versus} \quad H_{11,i}: \alpha_i\beta_i \neq 0, \quad i=1,2,\dots,m.
\end{equation}
Denote by $\boldsymbol{\xi}_i = (\xi_{i1}, \xi_{i2})$ the latent vector indicating the underlying truth of the $i$-th hypothesis, where $\xi_{i1} = \mathbf{1}\{\alpha_i \neq 0\}$ and $\xi_{i2} = \mathbf{1}\{\beta_i \neq 0\}$. The vector $\boldsymbol{\xi}_i$ takes values in the set $\{(0,0), (0,1), (1,0), (1,1)\}$. Let $m_{jk} = \sum^{m}_{i=1}\mathbf{1}\{\boldsymbol{\xi}_i=(j,k)\}$ represent the count of hypotheses in each state for $0\leq j,k\leq 1$.

We assume a prior distribution $P(\boldsymbol{\xi}_i=(j,k)) = \pi_{jk}$, where $\pi_{jk} \geq 0$ and $\sum_{j,k} \pi_{jk} = 1$. Conditional on the latent states, the non-zero effect sizes are modeled as Gaussian:
\[
\sqrt{n}\alpha_i \mid (\xi_{i1} = 1) \sim N(\mu, \psi) \quad \text{and} \quad \sqrt{n}\beta_i \mid (\xi_{i2} = 1) \sim N(\theta, \kappa).
\]
When $\xi_{i1} = 0$, $\alpha_i = 0$ (and analogously for $\beta_i$). Assuming that $\alpha_i$ and $\beta_i$ are independent conditional on $\boldsymbol{\xi}_i$, the marginal distribution of the coefficient estimates $(\hat{\alpha}_i, \hat{\beta}_i)$ given the latent states is:
\begin{align}\label{GMM-o}
\sqrt{n}
\begin{pmatrix}
\hat{\alpha}_i \\
\hat{\beta}_i
\end{pmatrix}
\,\Bigg|\, \boldsymbol{\xi}_i
\sim N\left(
\begin{pmatrix}
\mu\xi_{i1} \\
\theta\xi_{i2}
\end{pmatrix},
\begin{pmatrix}
\sigma_{i1}^2 + \psi\xi_{i1} & 0 \\
0 & \sigma_{i2}^2 + \kappa\xi_{i2}
\end{pmatrix}
\right),
\end{align}
where $(\sigma_{i1}^2/n, \sigma_{i2}^2/n)$ denote the variances of the estimates $(\hat{\alpha}_i, \hat{\beta}_i)$ conditional on $(\alpha_i, \beta_i)$.

We employ an Expectation-Maximization (EM) algorithm to estimate the unknown parameters $\boldsymbol{\Theta} = \{ \boldsymbol{\pi}, \mu, \theta, \psi, \kappa \}$; details are provided in Section 2 of the Supplementary Materials.
Although MLFDR is derived under the framework of Model (\ref{sem}), it is applicable to a broader range of settings, including binary outcomes, mediator-exposure interactions, and models with confounders. The extensive simulations presented earlier demonstrate the method's robustness across these varied scenarios.

\subsection{Extension: Composite Alternative}

In the last section we introduced the latent variable $\boldsymbol{\xi}_i = (\xi_{i1}, \xi_{i2})$ to characterize the underlying state of the hypothesis. We now extend this framework to a composite alternative setting, where the non-zero effects are drawn from mixture distributions. Specifically, we assume $(\alpha_i, \beta_i)$ are generated as follows:
\begin{align*}
    \sqrt{n}\alpha_i \mid \{\xi_{i1} = u\} &\sim N(\mu_u, \kappa_u),\quad u=0,1,\dots,d_1, \\
    \sqrt{n}\beta_i \mid \{\xi_{i2} = v\} &\sim N(\theta_v, \psi_v),\quad v=0,1,\dots,d_2.
\end{align*}
Here, the index $0$ denotes the null state, such that $\mu_0 = \theta_0 = \kappa_0 = \psi_0 = 0$ (i.e., a degenerate distribution at zero). As before, $\alpha_i$ and $\beta_i$ are assumed independent conditional on the latent state $\boldsymbol{\xi}_i$.

Marginalizing over the prior distribution, the joint distribution of the estimators $(\hat{\alpha}_i, \hat{\beta}_i)$ follows a Gaussian Mixture Model (GMM):
\begin{align}\label{joint-GMM}
\sqrt{n}
\begin{pmatrix}
\hat{\alpha}_i \\
\hat{\beta}_i
\end{pmatrix}
\sim \sum^{d_1}_{u=0}\sum^{d_2}_{v=0}\pi_{uv} N\left(
\begin{pmatrix}
        \mu_u \\
        \theta_v
       \end{pmatrix},
\begin{pmatrix}
       \sigma_{i1}^2 +\kappa_u & 0 \\
       0 & \sigma_{i2}^2 + \psi_v
     \end{pmatrix}
     \right),   
\end{align}
where $\pi_{uv} = P(\xi_{i1}=u, \xi_{i2}=v)$. In this context, the component null hypotheses map to specific index combinations: $H_{01,i}$ corresponds to $\{\xi_{i1}=0, \xi_{i2} \in \{1,\dots,d_2\}\}$, while the alternative $H_{11,i}$ corresponds to $\{\xi_{i1} \in \{1,\dots,d_1\}, \xi_{i2} \in \{1,\dots,d_2\}\}$.

The corresponding marginal distributions are given by:
\begin{equation}\label{marginal-GMM}
    \begin{aligned}
    \sqrt{n}\hat{\alpha}_i &\sim \sum_{u = 0}^{d_1}\pi_{u\cdot}N(\mu_u, \sigma^2_{i1} + \kappa_u),\\
    \sqrt{n}\hat{\beta}_i &\sim \sum_{v = 0}^{d_2}\pi_{\cdot v} N(\theta_v, \sigma^2_{i2} + \psi_v),  
    \end{aligned}
\end{equation}
where $\pi_{u\cdot} = \sum_v \pi_{uv}$ and $\pi_{\cdot v} = \sum_u \pi_{uv}$.

\subsubsection*{Two-Step EM Algorithm}
Directly fitting a $(d_1 + 1)(d_2 + 1)$-component bivariate GMM to estimate the joint probability matrix $\boldsymbol{\pi}$ is computationally intensive. To mitigate this burden, we propose a two-step Expectation-Maximization (EM) algorithm. 

In the first step, we estimate the parameters $\{\mu_u, \kappa_u\}_u$ and $\{\theta_v, \psi_v\}_v$ using the univariate marginal distributions described in Equation (\ref{marginal-GMM}). While the marginal mixing proportions $\pi_{u\cdot}$ and $\pi_{\cdot v}$ are insufficient to recover the joint distribution $\pi_{uv}$ (without assuming independence), the moment estimates remain valid. 

In the second step, we fix these mean and variance estimates and fit a constrained bivariate GMM to the joint data solely to estimate the joint mixing proportions $\pi_{uv}$. This approach significantly reduces the dimensionality of the optimization problem. Details of this algorithm are provided in Section 3 of the Supplementary Materials, and a supporting simulation is presented in Section \ref{sec:extension}.

This two-step approach is also applicable to the standard case where $d_1=d_2=1$. In Section 4 of the Supplementary Materials, we compare the standard bivariate EM against this two-step variant. The results indicate that while the two-step method offers substantial computational speedups, it incurs a slight reduction in power. Therefore, we recommend using the standard bivariate EM for simple cases ($d_1 = d_2 = 1$) and reserving the two-step EM for complex composite alternatives where computational efficiency is paramount.


 

\subsection{Step-Up Procedure Based on Local FDR}

\citet{Sun2007} demonstrated that for a simple null hypothesis, the oracle local FDR-based rejection region outperforms p-value-based thresholding. Specifically, it achieves a lower marginal False Non-discovery Rate (mFNR) while controlling the marginal False Discovery Rate (mFDR) at the same level. This advantage is particularly pronounced when the alternative distribution is asymmetric about the null. Motivated by these findings, we implement a local FDR-based step-up procedure to identify significant mediation pathways. 

We focus our discussion on the Gaussian Mixture Model described in (\ref{GMM-o}). The extension to the more general mixture model in (\ref{joint-GMM}) is straightforward. Conditional on the latent variables, the density function of the transformed statistics $\sqrt{n}(\hat{\alpha}_i,\hat{\beta}_i)$ under state $H_{jk,i}$ is given by:
\begin{equation}\label{mixture-pdfs}
\begin{aligned}
    f_{jk}(\cdot,\cdot) &:= \phi(\cdot,\cdot; \mu\delta_{j}, \theta \delta_{k}, \sigma_{i1}^2 + \kappa \delta_j, \sigma_{i2}^2 + \psi\delta_k),\\
    f &:= \sum_{j=0}^1 \sum_{k = 0}^1\pi_{jk}f_{jk},
\end{aligned}
\end{equation}
where $\delta_0=0$, $\delta_1=1$, and $\phi(\cdot,\cdot;\mu_1, \mu_2, \sigma_1^2, \sigma_2^2)$ denotes the bivariate normal density with mean vector $(\mu_1, \mu_2)$, variances $(\sigma_1^2, \sigma_2^2)$, and zero correlation.

For notational simplicity, let $a_i = \sqrt{n}\widehat{\alpha}_i$ and $b_i = \sqrt{n}\widehat{\beta}_i$. Under the composite null hypothesis, the joint local FDR is defined as:
\begin{align}
    \label{lfdr-1}
    \text{lfdr}(a_i, b_i) = \frac{\pi_{00}f_{00}(a_i, b_i) + \pi_{10}f_{10}(a_i, b_i) + \pi_{01}f_{01}(a_i, b_i)}{f(a_i, b_i)}.
\end{align}
This quantity can be estimated by substituting the parameter estimates obtained from the EM algorithm; we denote this estimate by $\widehat{\text{lfdr}}$. As the local FDR represents the posterior probability of the null hypothesis, a lower value indicates stronger evidence against the null. Consequently, we define the rejection region as $\widehat{\text{lfdr}}(a_i, b_i) \leq \delta$, where the threshold $\delta$ must be determined to control the error rate.

We assume the cumulative distribution function (CDF) of the joint local FDR is given by:
\begin{equation}\label{lfdr-dist}
    G(t) = \pi_{00}G_{00}(t) + \pi_{10}G_{10}(t)+ \pi_{01}G_{01}(t) + \pi_{11}G_{11}(t),
\end{equation}
where $G_{jk}(t)$ is the conditional CDF of $\text{lfdr}(a_i, b_i)$ under hypothesis $H_{jk}$:
\begin{align}
    G(t) &= \frac{1}{m}\sum_{i = 1}^m\mathbb{P}(\text{lfdr}(a_i, b_i) \leq t),\\
    G_{jk}(t) &= \frac{1}{m}\sum_{i = 1}^m\mathbb{P}(\text{lfdr}(a_i, b_i) \leq t \mid H_{jk}), \quad j,k \in \{0,1\}.
\end{align}

\subsubsection*{Oracle Procedure}\label{oracle}
The oracle procedure assumes that all parameters $\boldsymbol{\pi} = \{\pi_{00}, \pi_{10}, \pi_{01}, \pi_{11} \}$ and densities $f_{jk}$ are known. For a given threshold $\delta$, we define the following counting processes:
\begin{align}
    V_m(\delta) &= \sum_{i \in H_{00} \cup H_{10} \cup H_{01}} \mathbf{1}\{\text{lfdr}(a_i, b_i) \leq \delta\},\\
    R_m(\delta) &= \sum_{i = 1}^m \mathbf{1}\{\text{lfdr}(a_i, b_i) \leq \delta\},\\
    P_m(\delta) &= \sum_{i \in H_{11}} \mathbf{1}\{\text{lfdr}(a_i, b_i) > \delta\},\\
    W_m(\delta) &= \sum_{i = 1}^m \mathbf{1}\{\text{lfdr}(a_i, b_i) \leq \delta\} \cdot \text{lfdr}(a_i, b_i),
\end{align}
where $V_m(\delta)$ represents the number of false rejections, $R_m(\delta)$ the total number of rejections, and $P_m(\delta)$ the number of missed discoveries (false negatives).
We can express $V_m(\delta)$ as:
\begin{equation*}
    V_m(\delta) = \sum_{i = 1}^m \mathbf{1}\{\text{lfdr}(a_i, b_i) \leq \delta\} \cdot \mathbf{1}\{\boldsymbol{\xi_i} \in \{(0,0),(1,0),(0,1)\} \}.
\end{equation*}
Taking the expectation, we obtain:
\begin{equation}
    \begin{aligned}
    E[V_m(\delta)] &= \sum_{i = 1}^m E\left[\mathbf{1}\{\text{lfdr}(a_i, b_i) \leq \delta\} \cdot \mathbf{1}\{\boldsymbol{\xi_i} \in \{(0,0),(1,0),(0,1)\} \}\right]\\
    &= m_{00}G_{00}(\delta) + m_{10}G_{10}(\delta)+ m_{01}G_{01}(\delta).
\end{aligned}
\end{equation}
The mFDR at threshold $\delta$ is defined as:
\begin{equation}\label{mFDR}
    \widetilde{Q}(\delta) = \frac{E[V_m(\delta)]}{E[R_m(\delta)]} = \frac{\pi_{00}G_{00}(\delta) + \pi_{10}G_{10}(\delta)+ \pi_{01}G_{01}(\delta)}{G(\delta)}.
\end{equation}
This quantity can be empirically estimated by:
\begin{equation}\label{mFDR-emp}
    Q_m(\delta) = \frac{\sum_{i = 1}^m \mathbf{1}\{\text{lfdr}(a_i, b_i) \leq \delta\}\text{lfdr}(a_i, b_i)}{\sum_{i = 1}^m \mathbf{1}\{\text{lfdr}(a_i, b_i) \leq \delta\}}.
\end{equation}
In Proposition 1 of the Supplementary Materials, we prove that for a fixed $\delta$, the numerator and denominator of $Q_m(\delta)$ are unbiased estimators of the numerator and denominator of $\widetilde{Q}(\delta)$, respectively. The oracle rejection region for the composite null is defined as $\mathbf{1}\{\text{lfdr}(a_i, b_i) \leq \delta_m\}$, where the threshold is selected as:
\begin{equation}
    \delta_m = \sup\{t \in (0,1): Q_m(t) \leq \alpha\}.
\end{equation}

\subsubsection*{Adaptive Procedure}\label{adapt}

In practical applications, the true parameters $\boldsymbol{\pi} = \{\pi_{00}, \pi_{10}, \pi_{01}, \pi_{11} \}$ and the component densities $f_{jk}$ are unknown. Consequently, the true local FDR values in Equation (\ref{mFDR-emp}) are inaccessible. To address this, we substitute the unknown quantities with their estimates obtained via the EM algorithm. By plugging in these estimates, we obtain the estimated local FDR, denoted as $\widehat{\text{lfdr}}$, which yields an empirical estimate of $Q_m(\delta)$:
\begin{equation}\label{mFDR-est}
    \widehat{Q}_m(\delta) = \frac{\sum_{i = 1}^m \mathbf{1}\{\widehat{\text{lfdr}}(a_i, b_i) \leq \delta\}\widehat{\text{lfdr}}(a_i, b_i)}{\sum_{i = 1}^m \mathbf{1}\{\widehat{\text{lfdr}}(a_i, b_i) \leq \delta\}}.
\end{equation}
Accordingly, the data-adaptive threshold $\widehat{\delta}_m$ is defined as:
\[
\widehat{\delta}_m = \sup\{t \in (0,1): \widehat{Q}_m(t) \leq \alpha\}.
\]
The resulting procedure, which rejects the $i$-th hypothesis if $\widehat{\text{lfdr}}(a_i, b_i) \leq \widehat{\delta}_m$, is operationally equivalent to the step-up algorithm detailed in Algorithm \ref{alg:stepup}.
\begin{algorithm}[t]
    \caption{A data-adaptive procedure for finding the cutoffs in MLFDR}
  \label{alg:stepup}
  \begin{algorithmic}[1]
 \Require EM estimates for the parameters $\Theta = \{\pi, \mu, \theta, \psi, \kappa\}$
 \For{$i = 1, \dots, m$}
 \State  Compute $\widehat{\text{lfdr}}_i = \widehat{\text{lfdr}}(\sqrt{n}\hat{\alpha}_i, \sqrt{n}\hat{\beta}_i)$
 \EndFor
 \State Sort the estimated local FDR values in ascending order to obtain the set of order statistics $\Psi = \{\widehat{\text{lfdr}}_{(1)}, \widehat{\text{lfdr}}_{(2)}, \dots, \widehat{\text{lfdr}}_{(m)}\}$, where ties are broken at random.
 \State For each $k \in \{1, \dots, m\}$, compute the average estimated local FDR for the top $k$ statistics:
    \[
    \widehat{Q}_m(\widehat{\text{lfdr}}_{(k)}) = \frac{1}{k}\sum_{i = 1}^k\widehat{\text{lfdr}}_{(i)}.
    \]

\State Determine the cutoff index $k = \max\left\{j: \frac{1}{j}\sum_{i = 1}^j\widehat{\text{lfdr}}_{(i)} \leq \alpha\right\}$

\State Reject the null hypotheses corresponding to the smallest $k$ local FDR values, denoted as $H_{(1)}, H_{(2)}, \dots, H_{(k)}$
 
  \end{algorithmic}
\end{algorithm}

\section{Theoretical Properties}\label{theory}

To establish the asymptotic properties of MLFDR, we first define the limiting random variables $(\alpha_{i,0}, \beta_{i,0})$. Conditional on the latent state $\boldsymbol{\xi}_i=(\xi_{i1},\xi_{i2})$, these variables follow a Gaussian Mixture Model (GMM):
\begin{align}\label{GMM}
\left(\alpha_{i,0}, \beta_{i,0} \right) \sim N\left(\begin{pmatrix}
        \mu\xi_{i1} \\
        \theta\xi_{i2}
       \end{pmatrix},
\begin{pmatrix}
       \sigma_{i1,0}^2 +\kappa\xi_{i1} & 0 \\
       0 & \sigma_{i2,0}^2 + \psi\xi_{i2}
     \end{pmatrix}
     \right),
\end{align}
where the limiting variances are defined as $\sigma_{i1,0}^2=\sigma^2_{i,a}/\mathbb{E}[X^2]$ and $\sigma_{i2,0}^2=\sigma^2_{i,b} \mathbb{E}[X^2]/\{\mathbb{E}[X^2]\mathbb{E}[M^2_i]-(\mathbb{E}[XM_i])^2\}$.

Let $G^0(t)$ and $G^0_{jk}(t)$ denote the cumulative distribution functions of the local FDR based on these limiting variables:
\begin{align*}
    G^0(t) &= \frac{1}{m}\sum_{i = 1}^m\mathbb{P}(\text{lfdr}(\alpha_{i,0}, \beta_{i,0}) \leq t),\\
    G^0_{jk}(t) &= \frac{1}{m}\sum_{i = 1}^m\mathbb{P}(\text{lfdr}(\alpha_{i,0}, \beta_{i,0}) \leq t \mid H^0_{jk}), \quad \text{for}\quad j,k \in \{0,1\}.
\end{align*}
We further define the limiting function $V(\delta) := \pi_{00}G^0_{00}(\delta) + \pi_{10}G^0_{10}(\delta)+ \pi_{01}G^0_{01}(\delta)$. Our theoretical analysis relies on establishing the following pointwise convergence results:
\begin{equation}\label{pointwise-conv}
    \begin{aligned}
        \frac{1}{m}V_m(\delta) \overset{P}{\to} V(\delta), \quad
    \frac{1}{m}W_m(\delta) \overset{P}{\to}V(\delta),\quad
    \frac{1}{m} R_m(\delta)\overset{P}{\to}G^0(\delta).
    \end{aligned}
\end{equation}
Additionally, we define the theoretical and empirical marginal FDR (mFDR) processes as:
\begin{align}\label{mfdr-oracle-est}
     Q(\delta) = \frac{V(\delta)}{G^0(\delta)},\quad \widetilde{Q}_m(\delta) = \frac{V_m(\delta)}{R_m(\delta)\lor 1}.
\end{align}
Note that at a cutoff $\delta$, the FDR for the oracle procedure is given by $\text{FDR}_m(\delta) = \mathbb{E}[\widetilde{Q}_m(\delta)]$.

We assume the following conditions to prove asymptotic FDR control:

\noindent
\textbf{(A1) EM estimates are AMLE.} We assume the estimators satisfy the Approximate Maximum Likelihood Estimator (AMLE) property. An estimator $\hat{f}$ is an AMLE of a density $f$ based on samples $(\alpha_1, \beta_1), \dots, (\alpha_m, \beta_m)$ if:
\[ \prod_{i = 1}^m \frac{\hat{f}(\alpha_i, \beta_i)}{f(\alpha_i, \beta_i)} \geq 1. \]
This condition, previously utilized by \cite{deb2022two}, ensures that the estimator yields a likelihood at least as high as the true parameters. This property is essential for establishing the convergence of the estimated GMM density to the truth under the Hellinger distance, thereby ensuring the identifiability of the EM algorithm.

\noindent
\textbf{(A2) Independence of samples.} The observations $\{M^j_i\}_{j=1}^n$ are mutually independent, as are $\{X_j\}_{j=1}^n$.

\noindent
\textbf{(A3) Sub-Gaussianity.} The variables satisfy the following tail conditions: $M_i \sim \text{subG}(\nu_i^2)$, $M_iX \sim \text{subG}(\eta_i^2)$, and $X \sim \text{subG}(\rho^2)$. Note that if $M_i \sim \text{subG}(\nu_i^2)$, it follows that $(M_i^2 - \mathbb{E}M_i^2) \sim \text{subE}(16\nu_i^2)$ and $(X^2 - \mathbb{E}X^2) \sim \text{subE}(16\rho^2)$, where subG and subE denote sub-Gaussian and sub-exponential variables, respectively. A variable $Z$ is sub-Gaussian with parameter $\rho$ if $\mathbb{P}(|Z| \geq t) \leq 2\exp(-\frac{t^2}{2\rho^2})$. A variable $Y$ is sub-exponential with parameter $\eta$ if 
$$\mathbb{P}(|Y - \mathbb{E}Y| > t) \leq 2\exp\left(-\frac{1}{2}\min\left(\frac{t}{\eta}, \frac{t^2}{\eta^2}\right)\right).$$

\noindent
\textbf{(A4) Finite variance.} There exist constants $c, C$ such that $0 < c < \min_i \sigma_{ij,0}^2 \leq \max_i \sigma_{ij,0}^2 < C < \infty$ for $j=1,2$.

\noindent
\textbf{(A5) Critical point.} There exists a $\delta_\infty \in (0,1]$ such that $Q(\delta_\infty) < \alpha$.

Assumptions \textbf{(A2)}--\textbf{(A4)} enable the derivation of the uniform convergence of $\sigma_{i1}^2$ and $\sigma_{i2}^2$ to their limiting counterparts in (\ref{GMM}) (Lemma 1). This variance convergence is subsequently used to prove the pointwise convergence results in Lemma 2. Assumption \textbf{(A5)}, similar to conditions used in \cite{Cao2021}, ensures the existence of a valid cutoff that asymptotically controls the FDR at level $\alpha$.

We state the main theoretical results below. Theorem \ref{main-thm1} establishes FDR control for the oracle rule, while Theorem \ref{main-thm2} extends this control to the adaptive procedure.

\begin{theorem}\label{main-thm1}
    Under Assumptions \textbf{(A2)} -- \textbf{(A5)}, if $m,n \to \infty$ such that $m = o(\exp(kn))$ for some constant $k$, then:
\[ \limsup_{m \to \infty} \text{FDR}_m(\delta_m) \leq \alpha.\]
\end{theorem}

\begin{theorem}\label{main-thm2}
    Under Assumptions \textbf{(A1)} -- \textbf{(A5)}, if $m,n \to \infty$ such that $m = o(\exp(kn))$ for some constant $k$, then:
\[ \limsup_{m \to \infty} \text{FDR}_m(\widehat{\delta}_m) \leq \alpha.\]
\end{theorem}
Proofs and ancillary results are provided in the Supplementary Materials.

\section{Supplementary Materials}
\subsection{Properties of the Step-Up Procedure}
\begin{proposition}\label{stepup-1}
Let $G_{00}, G_{10}, G_{01}$ be the cumulative distribution functions defined previously. The following identity holds:
\begin{align*}
    \mathbb{E}&\left[\frac{1}{m} \sum_{i = 1}^m \mathbf{1}\{\text{lfdr}(\sqrt{n}\hat{\alpha}_i, \sqrt{n}\hat{\beta}_i) \leq \delta\}\,\text{lfdr}(\sqrt{n}\hat{\alpha}_i, \sqrt{n}\hat{\beta}_i)\right] \\
    = &\pi_{00}G_{00}(\delta) + \pi_{10}G_{10}(\delta)+ \pi_{01}G_{01}(\delta).
\end{align*}
\end{proposition}
\begin{proof}
    Let $a_i = \sqrt{n}\hat{\alpha}_i$ and $b_i = \sqrt{n}\hat{\beta}_i$. Define the rejection region $S_i^n(\delta) = \{(a_i, b_i) : \text{lfdr}(a_i, b_i) \leq \delta \}$. By the linearity of expectation and the definition of the local FDR, we have:
\begin{align*}
    \mathbb{E}&\left[ \frac{1}{m}\sum_{i = 1}^m \mathbf{1}\{\text{lfdr}(a_i, b_i) \leq \delta\}\,\text{lfdr}(a_i, b_i)\right]\\
    &= \frac{1}{m}\sum_{i = 1}^m \mathbb{E}\left[\mathbf{1}\{\text{lfdr}(a_i, b_i) \leq \delta\}\,\text{lfdr}(a_i, b_i)\right]\\
    &= \frac{1}{m}\sum_{i = 1}^m \int_{S_i^n(\delta)}\text{lfdr}(a_i, b_i)f(a_i, b_i)\,d a_i \, d b_i\\
    &= \frac{1}{m}\sum_{i = 1}^m \int_{S_i^n(\delta)} \left( \sum_{j,k \in \{0,1\} \setminus \{1,1\}} \pi_{jk}f_{jk}(a_i,b_i) \right) \,d a_i \, d b_i\\
    &= \pi_{00}\frac{1}{m}\sum_{i = 1}^m \int_{S_i^n(\delta)} f_{00}(a_i,b_i)\,d a_i \, d b_i \\
    &\quad + \pi_{10}\frac{1}{m}\sum_{i = 1}^m \int_{S_i^n(\delta)} f_{10}(a_i,b_i)\,d a_i \, d b_i\\ 
    &\quad + \pi_{01}\frac{1}{m}\sum_{i = 1}^m \int_{S_i^n(\delta)} f_{01}(a_i,b_i)\,d a_i \, d b_i\\
    &= \pi_{00}\frac{1}{m}\sum_{i = 1}^m\mathbb{P}(\text{lfdr}(a_i, b_i) \leq \delta \mid H_{00})\\
    &\quad + \pi_{10}\frac{1}{m}\sum_{i = 1}^m\mathbb{P}(\text{lfdr}(a_i, b_i) \leq \delta \mid H_{10})\\ 
    &\quad + \pi_{01}\frac{1}{m}\sum_{i = 1}^m\mathbb{P}(\text{lfdr}(a_i, b_i) \leq \delta \mid H_{01})\\
    &= \pi_{00}G_{00}(\delta) + \pi_{10}G_{10}(\delta)+ \pi_{01}G_{01}(\delta).
    \end{align*}
\end{proof}
Based on Proposition \ref{stepup-1}, and observing that
\[
\mathbb{E}\left[\frac{1}{m} \sum_{i=1}^m \mathbf{1}\{ \text{lfdr}(a_i, b_i) \leq \delta\}\right] = G(\delta),
\]
it is natural to estimate the marginal FDR,
\[
\widetilde{Q}(\delta) = \frac{\pi_{00}G_{00}(\delta) + \pi_{10}G_{10}(\delta)+ \pi_{01}G_{01}(\delta)}{G(\delta)},
\]
using the empirical quantity:
\[
Q_m(\delta) = \frac{\sum_{i = 1}^m \mathbf{1}\{\text{lfdr}(a_i, b_i) \leq \delta\}\,\text{lfdr}(a_i, b_i)}{\sum_{i = 1}^m \mathbf{1}\{\text{lfdr}(a_i, b_i) \leq \delta\}}.
\]

\subsection{EM Algorithm Details}
Let $\Gamma = (\pi_{00}, \pi_{10}, \pi_{01}, \pi_{11}, \mu, \theta, \kappa, \psi)$ denote the set of unknown parameters. We define the component densities corresponding to the four latent states as follows:
\begin{equation}\label{references2}
\begin{aligned}
    f_{00}(\sqrt{n}\hat{\alpha}_i, \sqrt{n}\hat{\beta}_i) &:= \phi(\sqrt{n}\hat{\alpha}_i, \sqrt{n}\hat{\beta}_i; 0, 0, \sigma_{i1}^2, \sigma_{i2}^2),\\
    f_{10}(\sqrt{n}\hat{\alpha}_i, \sqrt{n}\hat{\beta}_i) &:= \phi(\sqrt{n}\hat{\alpha}_i, \sqrt{n}\hat{\beta}_i; \mu, 0, \sigma_{i1}^2 + \kappa, \sigma_{i2}^2),\\
    f_{01}(\sqrt{n}\hat{\alpha}_i, \sqrt{n}\hat{\beta}_i) &:= \phi(\sqrt{n}\hat{\alpha}_i, \sqrt{n}\hat{\beta}_i; 0, \theta, \sigma_{i1}^2, \sigma_{i2}^2 + \psi),\\
    f_{11}(\sqrt{n}\hat{\alpha}_i, \sqrt{n}\hat{\beta}_i) &:= \phi(\sqrt{n}\hat{\alpha}_i, \sqrt{n}\hat{\beta}_i; \mu, \theta, \sigma_{i1}^2 + \kappa, \sigma_{i2}^2 + \psi),
\end{aligned}
\end{equation}
where $\phi(\cdot, \cdot; \mu_1, \mu_2, \sigma_1^2, \sigma_2^2)$ denotes the bivariate normal density with mean vector $(\mu_1, \mu_2)$, variances $(\sigma_1^2, \sigma_2^2)$, and zero correlation. The OLS variance estimates $(\sigma_{i1}^2, \sigma_{i2}^2)$ are treated as fixed inputs in the algorithm.

The complete data log-likelihood is given by:
\begin{equation}\label{llk}
    \begin{aligned}
    L(\Gamma) &= \sum_{i = 1}^m \sum_{u,v \in \{0,1\}} \mathbbm{1}(\boldsymbol{\xi}_i = (u,v)) \log f_{uv}(\sqrt{n}\hat{\alpha}_i, \sqrt{n}\hat{\beta}_i) \\
    &\quad + \sum_{i = 1}^m \sum_{u,v \in \{0,1\}} \mathbbm{1}(\boldsymbol{\xi}_i = (u,v)) \log \pi_{uv}.
    \end{aligned}
\end{equation}

\subsubsection*{E-Step}
In the Expectation step, we calculate the conditional probability (posterior) of the latent variable $\boldsymbol{\xi}_i$ given the observed data $(\hat{\alpha}_i, \hat{\beta}_i)$ and the current parameter estimates $\Gamma^{(t)}$:
\begin{align*}
    Q_{i;u,v}^{(t)} &:= \mathbb{P}(\boldsymbol{\xi}_i = (u,v) \mid \hat{\alpha}_i, \hat{\beta}_i, \Gamma^{(t)}) \\
    &= \frac{\pi^{(t)}_{uv} f_{uv}^{(t)}(\sqrt{n}\hat{\alpha}_i, \sqrt{n}\hat{\beta}_i)}{\sum_{j,k \in \{0,1\}} \pi^{(t)}_{jk} f_{jk}^{(t)}(\sqrt{n}\hat{\alpha}_i, \sqrt{n}\hat{\beta}_i)}, \quad u,v \in \{0,1\},
\end{align*}
where $f_{uv}^{(t)}$ denotes the density evaluated using the parameters from iteration $t$. The expected complete data log-likelihood is then:
\begin{align*}
    \mathcal{Q}(\Gamma \mid \Gamma^{(t)}) &= \sum_{i = 1}^m \sum_{u,v \in \{0,1\}} Q_{i;u,v}^{(t)} \log f_{uv}(\sqrt{n}\hat{\alpha}_i, \sqrt{n}\hat{\beta}_i) \\
    &\quad + \sum_{i = 1}^m \sum_{u,v \in \{0,1\}} Q_{i;u,v}^{(t)} \log \pi_{uv}.
\end{align*}

\subsubsection*{M-Step}
We maximize $\mathcal{Q}(\Gamma \mid \Gamma^{(t)})$ with respect to $\Gamma$ to obtain updated estimates.

\textbf{1. Updating Mixing Proportions ($\pi$):}
Maximizing with respect to $\pi_{uv}$ under the constraint $\sum_{u,v} \pi_{uv} = 1$ yields the standard closed-form update:
\begin{align}\label{pi_est}
     \pi_{uv}^{(t+1)} = \frac{1}{m}\sum_{i = 1}^m Q_{i;u,v}^{(t)}, \quad u,v \in \{0,1\}.
\end{align}

\textbf{2. Updating Means ($\mu, \theta$):}
Setting the partial derivative with respect to $\mu$ to zero:
\[
\frac{\partial \mathcal{Q}}{\partial \mu} = \sum_{i = 1}^m Q_{i;1,0}^{(t)}\frac{\sqrt{n}\hat{\alpha}_i - \mu}{\sigma_{i1}^2 + \kappa^{(t)}} + \sum_{i = 1}^m Q_{i;1,1}^{(t)}\frac{\sqrt{n}\hat{\alpha}_i - \mu}{\sigma_{i1}^2 + \kappa^{(t)}} = 0.
\]
Solving for $\mu$ yields the weighted least squares estimator:
\begin{align*}
    \mu^{(t+1)} = \sqrt{n}\frac{\sum_{i = 1}^m w_i^{(t)} \hat{\alpha}_i}{\sum_{i = 1}^m w_i^{(t)}}, \quad \text{where } w_i^{(t)} = \frac{Q_{i;1,0}^{(t)} + Q_{i;1,1}^{(t)}}{\sigma_{i1}^2 + \kappa^{(t)}}.
\end{align*}
Similarly, maximizing with respect to $\theta$ yields:
\begin{align*}
    \theta^{(t+1)} = \sqrt{n}\frac{\sum_{i = 1}^m v_i^{(t)}\hat{\beta}_i}{\sum_{i = 1}^m v_i^{(t)}}, \quad \text{where } v_i^{(t)} = \frac{Q_{i;0,1}^{(t)} + Q_{i;1,1}^{(t)}}{\sigma_{i2}^2 + \psi^{(t)}}.
\end{align*}

\textbf{3. Updating Variances ($\kappa, \psi$):}
Due to the heterogeneity of $\sigma_{i1}^2$ and $\sigma_{i2}^2$, closed-form updates for $\kappa$ and $\psi$ are not available. Therefore, the estimates $\kappa^{(t+1)}$ and $\psi^{(t+1)}$ are obtained by performing a 1-D numerical grid search on the likelihood function given in (\ref{llk}).

\subsection{Two-Step EM Algorithm}

\textbf{Step 1: Marginal Parameter Estimation}

In the first step, we fit a $(d_1 + 1)$-component univariate GMM to the statistics $\sqrt{n}\hat{\alpha}_i$:
\[
\sqrt{n}\hat{\alpha}_i \sim \sum_{u = 0}^{d_1}\pi_{u\cdot} N(\mu_u, \sigma^2_{i1} + \kappa_u).
\]
The posterior probability of component $u$ given the $i$-th observation at iteration $t$ is:
\[
Q^{(t)}_{iu} = \frac{\pi^{(t)}_{u\cdot} \phi(\sqrt{n}\hat{\alpha}_i; \mu^{(t)}_u, \sigma^2_{i1} + \kappa^{(t)}_u)}{\sum_{k = 0}^{d_1} \pi^{(t)}_{k\cdot} \phi(\sqrt{n}\hat{\alpha}_i; \mu^{(t)}_k, \sigma^2_{i1} + \kappa^{(t)}_k)}.
\]
The mean parameters $\mu_u$ are updated via weighted least squares:
\[
\mu^{(t+1)}_u = \sqrt{n}\frac{\sum_{i = 1}^m w^{(t)}_{iu} \hat{\alpha}_i}{\sum_{i = 1}^m w^{(t)}_{iu}}, \quad \text{where } w^{(t)}_{iu} = \frac{Q^{(t)}_{iu}}{\sigma^2_{i1} + \kappa^{(t)}_u}.
\]
The variance parameters $\kappa_u$ are estimated via a 1-D grid search maximizing the marginal log-likelihood. The parameters $\{\theta_v, \psi_v\}_{v=1}^{d_2}$ corresponding to $\sqrt{n}\hat{\beta}_i$ are estimated analogously.

\noindent
\textbf{Step 2: Joint Mixing Proportion Estimation}

We fix the parameters $\{\hat{\mu}_u, \hat{\kappa}_u\}_{u = 1}^{d_1}$ and $\{\hat{\theta}_v, \hat{\psi}_v\}_{v = 1}^{d_2}$ obtained from Step 1. Let the latent state be denoted by $\boldsymbol{\xi_i} = (\xi_{i1}, \xi_{i2})$, where $\xi_{i1} \in \{0, \dots, d_1\}$ and $\xi_{i2} \in \{0, \dots, d_2\}$. Our goal is to estimate the joint probability matrix $\pi_{uv} = \mathbb{P}(\xi_{i1} = u, \xi_{i2} = v)$.
Using the fixed marginal parameters, we compute the conditional probability of the latent state $\boldsymbol{\xi_i}$ given the data:
\[
P^{(t)}_{i;u,v} = \frac{\pi^{(t)}_{uv} \phi(\sqrt{n}\hat{\alpha}_i, \sqrt{n}\hat{\beta}_i \mid \xi_{i1} = u, \xi_{i2} = v)}{\sum_{j = 0}^{d_1}\sum_{k = 0}^{d_2} \pi^{(t)}_{jk} \phi(\sqrt{n}\hat{\alpha}_i, \sqrt{n}\hat{\beta}_i \mid \xi_{i1} = j, \xi_{i2} = k)}.
\]
The joint mixing proportions are updated as:
\[
\pi^{(t+1)}_{uv} = \frac{1}{m}\sum_{i = 1}^m P^{(t)}_{i;u,v}.
\]
We iterate this procedure until convergence. Finally, the aggregate probabilities for the composite hypotheses are computed by summing over the relevant indices:
\begin{align*}
    \mathbb{P}(H_{00}) &= \pi_{00}, \\
    \mathbb{P}(H_{10}) &= \sum_{u = 1}^{d_1}\pi_{u0},\\
    \mathbb{P}(H_{01}) &= \sum_{v = 1}^{d_2}\pi_{0v},\\
    \mathbb{P}(H_{11}) &= \sum_{u = 1}^{d_1}\sum_{v = 1}^{d_2}\pi_{uv}.
\end{align*}

\subsection{Comparison: Standard EM vs. Two-Step EM}

The two-step EM algorithm is designed primarily to mitigate the computational burden associated with increasing component counts ($d_1, d_2$). To evaluate the trade-offs between the standard bivariate EM (MLFDR) and the two-step variant (MLFDR-2), we conducted a simulation with $d_1 = d_2 = 1$. We considered sample sizes of $n = 100$ and $300$ for $m = 1000$ under the Dense alternative scenarion, i.e $\boldsymbol{\pi} = \{0.4, 0.2, 0.2, 0.2 \}$

The data were generated as follows:
\begin{align*}
    X &\sim \text{Ber}(0.1),\\
    M_i &= \alpha_iX + e_i,\\
    Y_i &= \beta_iM_i + \gamma_iX + \epsilon_i.
\end{align*}
Figure \ref{em-comp} compares the empirical FDR and power of HDMT, DACT, MLFDR (using standard bivariate EM), and MLFDR-2 (using two-step EM).

We also benchmarked the computation times for comparable dimensions, scaling up to $m = 50,000$ and $n = 500$ (Table \ref{computation_time}). The results demonstrate that while the standard EM is efficient for smaller datasets, the two-step EM scales significantly better for high-dimensional data. For instance, at $m=50,000$, MLFDR-2 is approximately five times faster than the standard MLFDR.

\begin{table}[H]
    \centering
    \begin{tabular}{lccc}
    \hline
    Method & \multicolumn{1}{c}{$m = 1000$} & \multicolumn{1}{c}{$m = 5000$} & \multicolumn{1}{c}{$m = 50,000$} \\
           & \multicolumn{1}{c}{$n = 100$}  & \multicolumn{1}{c}{$n = 200$}  & \multicolumn{1}{c}{$n = 500$} \\
    \hline
    MLFDR      & 247 ms  & 1006 ms & 6253 ms \\
    MLFDR-2    & 1209 ms & 1246 ms & 1166 ms \\
    \hline
    \end{tabular}
    \caption{Computation times (in milliseconds) for MLFDR vs. Two-Step MLFDR across varying dimensions.}
    \label{computation_time}
\end{table}

\begin{figure}
    \centering
    \includegraphics[width = 0.7\linewidth]{EM_comparison.png}
    \caption{Comparison of empirical FDR and power across methods based on 100 simulation runs. The results distinguish between the standard MLFDR (Maroon) and the Two-Step MLFDR-2 (yellow).}
    \label{em-comp}
\end{figure}
As shown in Figure \ref{em-comp}, both MLFDR and MLFDR-2 offer power improvements over HDMT and DACT. However, the two-step approximation results in a reduction in power compared to the exact bivariate EM. Given that this marginal power loss could translate to missed discoveries in large-scale genomic studies, we recommend using the standard bivariate GMM when the model complexity is low (e.g., $d_1 + d_2 \leq 3$). For more complex models where computational resources are a constraint, the two-step EM provides a viable and efficient alternative.

\subsection{Theory}
\begin{lemma}\label{lemma-var}
    Under Assumptions \textbf{(A2)} - \textbf{(A4)}, if $m,n \to \infty$ such that $m = o(\text{exp}(kn))$ for all $ 0<k < \infty$ , then
$$  \max_{1\leq i\leq m}|\sigma_{i1}^2-\sigma_{i1,0}^2|\rightarrow^p 0,$$
and 
$$ \max_{1\leq i\leq m}|\sigma_{i2}^2-\sigma_{i2,0}^2|\rightarrow^p 0.$$
\end{lemma}
\begin{proof}
Define $D_n := n^{-1}\sum_{j = 1}^n (X^j)^2$ and $D := \mathbb{E}[X^2]$. Since $D > 0$, there exists a real number $a$ such that $D > 2a > 0$. Under Assumption \textbf{(A3)}, $X$ is sub-Gaussian, which implies that $X^2$ is sub-exponential. Applying Bernstein's inequality to the centered term $D_n - D$, we have:
\begin{align*}
    \sum_{i = 1}^m \mathbb{P}\left(\left|\sigma_{i1}^2-\sigma_{i1,0}^2\right| > \epsilon\right)
    &= \sum_{i = 1}^m\mathbb{P}\left(\left|\frac{1}{D_n} - \frac{1}{D}\right|>\delta_i\right) \\
    &\leq \sum_{i = 1}^m\mathbb{P}\left(\left|\frac{1}{D_n} - \frac{1}{D}\right|>\delta_i, D_n > a\right) + \sum_{i = 1}^m\mathbb{P}(D_n \leq a)\\
    &\leq \sum_{i = 1}^m\mathbb{P}\left(\left|D_n - D\right| > 2a^2\delta_i\right)+ \sum_{i = 1}^m\mathbb{P}(|D_n - D| \geq a)\\
    &\leq 2\sum_{i = 1}^m \exp\left(-\frac{n}{2}\left(\frac{4a^4\delta_i^2}{256\rho^4}\wedge\frac{2a^2\delta_i}{16\rho^2}\right)\right)\\
    &\quad + 2m \exp\left(-\frac{n}{2}\left(\frac{a^2}{256\rho^4}\wedge\frac{a}{16\rho^2}\right)\right),
\end{align*}
where $\delta_i = \epsilon/\sigma^2_{i,a}$.

Next, we consider the convergence of the mediator moments. By Assumption \textbf{(A3)}, $M_i^2$ is sub-exponential. Using Bernstein's inequality and the union bound, we obtain:
\begin{align*}
    \mathbb{P}\left(\max_{1\leq i \leq m}\left|n^{-1}\sum^{n}_{j=1}(M_i^j)^2 - \mathbb{E}[M_i^2]\right|>\epsilon\right) 
    &\leq \sum_{i = 1}^m \mathbb{P}\left(\left|n^{-1}\sum^{n}_{j=1}(M_i^j)^2 - \mathbb{E}[M_i^2]\right|>\epsilon\right)\\
    &\leq 2\sum_{i = 1}^m \exp\left(-\frac{n}{2}\left(\frac{\epsilon^2}{256\nu_i^4}\wedge\frac{\epsilon}{16\nu_i^2}\right)\right)\\
    &\rightarrow 0.
\end{align*}
Similarly, since $M_iX$ is sub-Gaussian by assumption, we apply Chernoff's bound and the union bound:
\begin{align*}
     \mathbb{P}\left(\max_{1\leq i \leq m}\left|n^{-1}\sum^{n}_{j=1}M_i^jX^j- \mathbb{E}[M_i X]\right|>\epsilon\right)
     &\leq 2\sum_{i = 1}^m \exp\left(-\frac{n\epsilon^2}{2\eta_i^2}\right)\\
    &\rightarrow 0.
\end{align*}

To analyze $\sigma_{i2}^2$, we introduce the following notation:
\begin{align*}
    A_{n,i} &:= \frac{1}{n}\sum^{n}_{j=1}(M_i^j)^2, \quad &A_i &:= \mathbb{E}[M_i^2],\\
    B_{n,i} &:= \frac{1}{n}\sum^{n}_{j=1}M_i^jX^j, \quad &B_i &:= \mathbb{E}[M_iX],\\
    C_{n} &:= \frac{1}{n}\sum^{n}_{j=1}X^j, \quad &C &:= \mathbb{E}[X].
\end{align*}
For a fixed $i$, the estimator $\sigma_{i2}^2$ is a continuous function of $(A_{n,i}, B_{n,i}, C_n)$ at the point $(A_i, B_i, C)$. By the definition of continuity, for any $\epsilon > 0$, there exists a $\delta > 0$ such that:
\begin{align*}
  \mathbb{P}\left( \max_{1\leq i\leq m}|\sigma_{i2}^2-\sigma_{i2,0}^2| > \epsilon\right) 
  &< \sum_{i = 1}^m \mathbb{P}(|\sigma_{i2}^2-\sigma_{i2,0}^2| > \epsilon)\\
  &\leq \sum_{i = 1}^m \mathbb{P}\left((A_{n,i} - A_i)^2 + (B_{n,i} - B_i)^2 + (C_n - C)^2 > \delta^2\right)\\
  &\leq \sum_{i = 1}^m \Big\{\mathbb{P}(|A_{n,i} - A_i| > \delta/\sqrt{3}) + \mathbb{P}(|B_{n,i} - B_i| > \delta/\sqrt{3}) \\
  &\quad + \mathbb{P}(|C_n - C| > \delta/\sqrt{3})\Big\}\\
  &\leq 2\sum_{i = 1}^m \exp\left(-\frac{n}{2}\left(\frac{\delta^2/3}{256\nu_i^4}\wedge\frac{\delta/\sqrt{3}}{16\nu_i^2}\right)\right) \\
  &\quad + 2\sum_{i = 1}^m \exp\left(-\frac{n\delta^2}{6\eta_i^2}\right) + m\exp\left(-\frac{n\delta^2}{6\rho^2}\right)\\
  &\rightarrow 0.
\end{align*}
Therefore, we conclude that:
\begin{equation}\label{sigma2-conv}
     \max_{1\leq i\leq m}|\sigma_{i2}^2-\sigma_{i2,0}^2|\xrightarrow{p} 0.
\end{equation}
\end{proof}

\begin{lemma} \label{process-convergence}
Let $V_m, W_m, R_m$ be defined as in Equations (14), (15), and (16). As $m, n \to \infty$, the following convergence results hold:
\begin{align}
    \frac{1}{m}V_m(\delta) &\xrightarrow{P} V(\delta), \\
    \frac{1}{m}W_m(\delta) &\xrightarrow{P} V(\delta),\\
    \frac{1}{m} R_m(\delta) &\xrightarrow{P} G^0(\delta).
\end{align}
\end{lemma}

\begin{proof}
    Conditional on $X$ and $\{M_i\}_i$, the local FDR statistics $\text{lfdr}(\sqrt{n}\hat{\alpha}_i, \sqrt{n}\hat{\beta}_i)$ are independent across $i$. Applying the Dvoretzky–Kiefer–Wolfowitz inequality \citep{naaman2021tight}, we have:
\begin{align*}
\mathbb{P}&\left(\sup_{\delta}\left|\frac{V_m(\delta) - \sum_{i \in H_0} \mathbb{P}(\text{lfdr}(\sqrt{n}\hat{\alpha}_i, \sqrt{n}\hat{\beta}_i) \leq \delta \mid X,\{M_i\}_i, H_{0,i})}{m}\right| >\epsilon \;\Bigg|\; X,\{M_i\}_i\right)
\\ &\leq d(m+1) \exp(-2m\epsilon^2).
\end{align*}
Unconditioning yields the unconditional bound:
\begin{align*}
\mathbb{P}&\left(\sup_{\delta}\left|\frac{V_m(\delta) - \sum_{i \in H_0}\mathbb{P}(\text{lfdr}(\sqrt{n}\hat{\alpha}_i, \sqrt{n}\hat{\beta}_i) \leq \delta \mid X,\{M_i\}_i, H_{0,i})}{m}\right| >\epsilon \right)\\
&\leq d(m+1) \exp(-2m\epsilon^2).   
\end{align*}

Define the conditional cumulative distribution function:
\begin{equation*}
    K_t(\sigma_{i1}^2, \sigma_{i2}^2) = \mathbb{P}(\text{lfdr}(\sqrt{n}\hat{\alpha}_i, \sqrt{n}\hat{\beta}_i) \leq t \mid X,\{M_i\}_i, H_{00}).
\end{equation*}
Since the distribution of $(\sqrt{n}\hat{\alpha}_i, \sqrt{n}\hat{\beta}_i)$ given $X,\{M_i\}_i$ differs from the limiting distribution $(\alpha_{i,0}, \beta_{i,0})$ only via the variances—$(\sigma_{i1}^2, \sigma_{i2}^2)$ versus $(\sigma_{i1,0}^2, \sigma_{i2,0}^2)$—we can identify:
\begin{equation*}
    K_t(\sigma_{i1,0}^2, \sigma_{i2,0}^2) = \mathbb{P}(\text{lfdr}(\alpha_{i,0}, \beta_{i,0}) \leq t \mid H^0_{00}).
\end{equation*}
The function $K_t(a,b)$ is continuous with respect to $(a,b)$. Therefore, applying the continuous mapping theorem and the convergence results from Lemma 1, we obtain:
\begin{align*}
    \mathbb{P}&(|G_{00}(t) - G^0_{00}(t)| > \epsilon) \\
    &\leq \mathbb{P}\left(\sum_{i =1}^m |K_t(\sigma_{i1}^2, \sigma_{i2}^2) - K_t(\sigma_{i1,0}^2, \sigma_{i2,0}^2)| > \epsilon m\right)\\
    &\leq \sum_{i =1}^m \mathbb{P}(|K_t(\sigma_{i1}^2, \sigma_{i2}^2) - K_t(\sigma_{i1,0}^2, \sigma_{i2,0}^2)| > \epsilon)\\
    &\leq \sum_{i =1}^m \mathbb{P}\left(|\sigma_{i1}^2-\sigma_{i1,0}^2|^2 + |\sigma_{i2}^2-\sigma_{i2,0}^2|^2 > \delta^2 \right)\\
    &\leq \sum_{i =1}^m \left[ \mathbb{P}\left(|\sigma_{i1}^2-\sigma_{i1,0}^2| > \frac{\delta}{\sqrt{2}}\right) + \mathbb{P}\left(|\sigma_{i2}^2-\sigma_{i2,0}^2| > \frac{\delta}{\sqrt{2}}\right)\right]\\
    &\to 0.
\end{align*}
By analogous arguments, we establish:
\begin{align*}
   G_{10}(t) &\xrightarrow{P} G^0_{10}(t),\\
   G_{01}(t) &\xrightarrow{P} G^0_{01}(t),\\
   G_{11}(t) &\xrightarrow{P} G^0_{11}(t).
\end{align*}
Combining these convergence results, it follows that:
\begin{align}
 \frac{1}{m} \sum_{i \in H_0}\mathbb{P}(\text{lfdr}(\sqrt{n}\hat{\alpha}_i, \sqrt{n}\hat{\beta}_i) \leq \delta \mid X,\{M_i\}_i, H_{0,i}) \xrightarrow{P} V(\delta).    
\end{align}
The convergence of $W_m(\delta)$ and $R_m(\delta)$ follows a similar derivation.
\end{proof}

\begin{proof}[Proof of Theorem 1]
    Invoking the Glivenko-Cantelli theorem, we establish the following uniform convergence results:
\begin{equation}\label{gliv-cantelli}
    \begin{aligned}
        \sup_{\delta \in [0,1]}\left|\frac{1}{m} R_m(\delta) - G^0(\delta)\right| &\xrightarrow{P} 0,\\
        \sup_{\delta \in [0,1]}\left|\frac{1}{m} V_m(\delta) - V(\delta)\right| &\xrightarrow{P} 0,\\
         \sup_{\delta \in [0,1]}\left|\frac{1}{m} W_m(\delta) - V(\delta)\right| &\xrightarrow{P} 0.
   \end{aligned}
\end{equation}
Define the realized False Discovery Proportion (FDP) process as:
\begin{align*}
   \widetilde{Q}_m(\delta) =  \frac{V_m(\delta)}{R_m(\delta)}.
\end{align*}
Following arguments analogous to Lemma 8.2 in \cite{Cao2021}, we have:
\begin{equation}\label{eq-1}
    \begin{aligned}
        \sup_{\delta \geq \delta_\infty}\left|Q_m(\delta) - Q(\delta) \right| &\xrightarrow{P} 0,\\
        \sup_{\delta \geq \delta_\infty}\left|\widetilde{Q}_m(\delta) - Q(\delta) \right| &\xrightarrow{P} 0.
    \end{aligned}
\end{equation}
The remainder of the proof follows the logic of Proposition 2.2 in \cite{Cao2021}. Define $e = \alpha - Q(\delta_\infty)$. By Assumption \textbf{(A5)}, we have $e > 0$. Combining this with (\ref{eq-1}), implies that $\mathbb{P}(|Q_m(\delta_\infty) - Q(\delta_\infty)| < e/2) \to 1$. Consequently, $\mathbb{P}(Q_m(\delta_\infty) < \alpha) \to 1$. By the definition of $\delta_m$, this implies that $\mathbb{P}(\delta_m \geq \delta_{\infty}) \to 1$.

We then bound the difference between the estimated and realized FDP at the threshold $\delta_m$:
\begin{align*}
    Q_m(\delta_m) - \widetilde{Q}_m(\delta_m) 
    &\geq \inf_{\delta \geq \delta_\infty}\{ Q_m(\delta) - \widetilde{Q}_m(\delta)\}\\
    &\geq \inf_{\delta \geq \delta_\infty}\{ Q_m(\delta) - Q(\delta) + Q(\delta) - \widetilde{Q}_m(\delta)\} \\
    &= o_\mathbb{P}(1).
\end{align*}
Therefore,
\begin{equation*}
    \widetilde{Q}_m(\delta_m) \leq Q_m(\delta_m) + o_\mathbb{P}(1) \leq \alpha + o_\mathbb{P}(1).
\end{equation*}
This implies:
\begin{equation*}
    \frac{V_m(\delta_m)}{R_m(\delta_m)\lor 1} \leq \frac{V_m(\delta_m)}{R_m(\delta_m)} =  \widetilde{Q}_m(\delta_m) \leq \alpha + o_\mathbb{P}(1).
\end{equation*}
Finally, applying Lemma 8.3 of \cite{Cao2021} ensures uniform integrability, yielding:
\begin{equation*}
    \limsup_{m \to \infty}\text{FDR}_m(\delta_m) = \limsup_{m \to \infty} \mathbb{E}\left[\frac{V_m(\delta_m)}{R_m(\delta_m)\lor 1}\right] \leq \alpha.
\end{equation*}
\end{proof}

\begin{proof}[Proof of Theorem 2]
    We define the empirical processes based on the estimated local FDR as:
\begin{align}
    \widehat{R}_m(\delta) &:= \sum_{i = 1}^m \mathbf{1}\{\widehat{\text{lfdr}}(\sqrt{n}\hat{\alpha}_i, \sqrt{n}\hat{\beta}_i) \leq \delta\},\\
    \widehat{W}_m(\delta) &:= \sum_{i = 1}^m \mathbf{1}\{\widehat{\text{lfdr}}(\sqrt{n}\hat{\alpha}_i, \sqrt{n}\hat{\beta}_i) \leq \delta\}\widehat{\text{lfdr}}(\sqrt{n}\hat{\alpha}_i, \sqrt{n}\hat{\beta}_i).
\end{align}
Following the argument in Lemma 8.4 of \cite{Cao2021} and invoking Lemma \ref{lfdrconv}, we establish the uniform convergence:
\begin{align}
    \sup_{\delta \geq \delta_\infty}\left|\frac{1}{m}\widehat{R}_m(\delta) - G^0(\delta) \right| &\xrightarrow{P} 0,\\
    \sup_{\delta \geq \delta_\infty}\left|\frac{1}{m}\widehat{W}_m(\delta) - V(\delta) \right| &\xrightarrow{P} 0.
\end{align}
Using arguments similar to Lemma 8.2 of \cite{Cao2021}, it follows that:
\[
\sup_{\delta \geq \delta_\infty}\left|\widehat{Q}_m(\delta) - Q(\delta)\right| \xrightarrow{P} 0.
\]
Proceeding analogously to Theorem 1, set $e = \alpha - Q(\delta_\infty) > 0$. We have:
\[
\left|\widehat{Q}_m(\delta_\infty) - Q(\delta_\infty) \right| \leq \sup_{\delta \geq \delta_\infty}\left|\widehat{Q}_m(\delta) - Q(\delta) \right| \leq e/2 \quad \text{with probability } \to 1.
\]
This implies $\mathbb{P}(\widehat{Q}_m(\delta_\infty) < \alpha) \to 1$, and consequently $\mathbb{P}(\widehat{\delta}_m \geq \delta_\infty ) \to 1$. Conditional on the event $\{\widehat{\delta}_m \geq \delta_\infty \}$, we have:
\begin{align*}
   \widetilde{Q}_m(\widehat{\delta}_m) - \widehat{Q}_m(\widehat{\delta}_m)
   &\leq \sup_{\delta \geq \delta_\infty}\left|\widehat{Q}_m(\delta) -  \widetilde{Q}_m(\delta)\right|\\
   &\leq \sup_{\delta \geq \delta_\infty}\left|\widehat{Q}_m(\delta) - Q(\delta) \right| + \sup_{\delta \geq \delta_\infty}\left|Q(\delta)- \widetilde{Q}_m(\delta)\right|\\
   &= o_\mathbb{P}(1).
\end{align*}
By the definition of $\widehat{\delta}_m$, we have $\widehat{Q}_m(\widehat{\delta}_m) \leq \alpha$. Therefore:
\begin{equation*}
    \frac{V_m(\widehat{\delta}_m)}{R_m(\widehat{\delta}_m)\lor 1} \leq \frac{V_m(\widehat{\delta}_m)}{R_m(\widehat{\delta}_m)} =  \widetilde{Q}_m(\widehat{\delta}_m) \leq \alpha + o_\mathbb{P}(1).
\end{equation*}
Applying Lemma 8.3 of \cite{Cao2021}, we conclude:
\begin{equation*}
    \limsup_{m \to \infty}\text{FDR}_m(\widehat{\delta}_m) = \limsup_{m \to \infty} \mathbb{E}\left[\frac{V_m(\widehat{\delta}_m)}{R_m(\widehat{\delta}_m)\lor 1}\right] \leq \alpha.
\end{equation*}
\end{proof}

\begin{defn}[Gaussian Mixture Models]
    Let $\mathcal{F}^2_{\text{Gauss}}$ denote the class of bivariate Gaussian Mixture Model (GMM) densities, defined as:
    $$\mathcal{F}^2_{\text{Gauss}} = \left\{f: f(\alpha, \beta) = \int \phi(\alpha, \beta; \mu, \theta, \sigma_1^2, \sigma_2^2, \rho)\,dG^*(\mu, \theta, \sigma_1^2, \sigma_2^2, \rho) \right\},$$
    where $G^*$ is a discrete probability measure supported on the set:
    \begin{align*}
        &\mu \in [-M_1, M_1],\\
        &\theta \in [-M_2, M_2],\\
        &\sigma_1^2, \sigma_2^2 \in (c, C),\\
        &\rho \in (0,1).
    \end{align*}
    We extend this definition to the class of $b$-variate GMM densities, denoted by $\mathcal{F}^b_{\text{Gauss}}$. Let $G^d_b$ represent the corresponding mixing distribution with at most $d$ atoms, defined on a bounded support. Specifically, if $\bfmu_i$ and $\bfS_i = (\sigma^i_{kl})_{k,l=1}^b$ denote the mean vector and the covariance matrix of the $i$-th component for $i = 1, \dots, d$, we assume:
    \begin{align*}
        \bfmu_i &\in [-M\boldsymbol{1}, M\boldsymbol{1}],\\
        \sigma^i_{kk} &\in [c, C].
    \end{align*}
    Note that this condition implies $\sigma^i_{kl} \in [-C, C]$ for $k \neq l$. Based on these parameters, we define $G^d_b$ as:
    $$G^d_b = \sum_{i = 1}^d \pi_i \delta(\bfmu_i, \bfS_i).$$
    If $\phi(\bfx, \bfmu, \bfS)$ denotes the density of a $b$-dimensional multivariate normal distribution with mean $\bfmu$ and covariance matrix $\bfS$, then:
    $$\mathcal{F}^b_{\text{Gauss}} = \left\{f: f(\bfx) = \sum_{i = 1}^d \pi_i\phi(\bfx, \bfmu_i, \bfS_i) = \int \phi(\bfx, \bfmu, \bfS)\,dG_b^d(\bfmu, \bfS) \right\}.$$
\end{defn}

\begin{defn}[AMLE]
    Let $\bfX$ be a $d$-dimensional random variable with density function $f: \mathbb{R}^d \to \mathbb{R}$ (or probability mass function for count data), where the parametric form of $f$ may be unknown. Based on a sample $(\bfX_1, \bfX_2, \dots, \bfX_m)$, an estimator $\hat{f}$ is defined as an Approximate Maximum Likelihood Estimator (AMLE) of $f$ if:
    $$\prod_{i = 1}^m \frac{\hat{f}(\bfX_i)}{f(\bfX_i)} \geq 1.$$
    That is, to qualify as an AMLE, the estimator must yield a likelihood evaluated at the observed samples that is at least as high as that of the true density.
\end{defn}

\begin{defn}[Hellinger Distance]
    We quantify the discrepancy between the estimated density $\hat{f}$ and the true density $f$ using the average squared Hellinger distance, given by:
    $$\mathcal{D}(\hat{f}, f) = \frac{1}{m}\sum_{i = 1}^m h^2(\hat{f}(\alpha_i, \beta_i),f(\alpha_i, \beta_i)),$$
    where
    $$h^2(\hat{f}(\alpha_i, \beta_i),f(\alpha_i, \beta_i)) = \frac{1}{2}\int \Big(\sqrt{\hat{f}(\alpha_i, \beta_i)} - \sqrt{f(\alpha_i, \beta_i)}\Big)^2\,d\alpha_i d\beta_i.$$ 
\end{defn}

\begin{defn}[Total Variation Distance]
    Alternatively, the loss between the estimated and true distributions can be quantified using the average Total Variation (TV) distance:
    $$TV(\hat{f}, f) = \frac{1}{m}\sum_{i = 1}^m \Delta(\hat{f}(\alpha_i, \beta_i),f(\alpha_i, \beta_i)),$$
    where
    $$\Delta(\hat{f}(\alpha_i, \beta_i),f(\alpha_i, \beta_i)) = \frac{1}{2}\int \left|\hat{f}(\alpha_i, \beta_i)-f(\alpha_i, \beta_i)\right|\,d\alpha_i d\beta_i.$$
    Note that the Hellinger and Total Variation distances satisfy the inequality:
    $$\mathcal{D}(\hat{f}, f) \leq TV(\hat{f}, f) \leq \sqrt{2}\sqrt{\mathcal{D}(\hat{f}, f)}.$$
\end{defn}

\begin{defn}[Restricted Supremum Norm]
    Let $h, \tilde{h}: \mathbb{R}^d \to \mathbb{R}$ be two functions. We define the supremum norm over a bounded subset $S \subset \mathbb{R}^d$ as:
    \begin{equation}\label{sup-norm}
        \norm{h - \tilde{h}}_{\infty, S} := \sup_{\bfx \in S}\bigg|h(\bfx) - \tilde{h}(\bfx)\bigg|.
    \end{equation}
\end{defn}

\begin{lemma}\label{coveringnum}
    Consider the family of $b$-variate Gaussian mixture distributions $\mathcal{F}^b_{\text{Gauss}}$ parameterized by the discrete mixing distributions $G^d_b$. Let $N(\eta, \mathcal{F}^b_{\text{Gauss}}, \|\cdot\|_{\infty, S})$ denote the covering number of this class under the restricted supremum norm defined in Equation (\ref{sup-norm}). Assuming that the covariance matrices of the component distributions are full-rank, the following bound holds:
    $$N(\eta, \mathcal{F}^b_{\text{Gauss}}, \|\cdot\|_{\infty, S}) \leq K \eta^{-\frac{2d + 3db + db^2}{2}},$$
    where $K$ is a constant depending on the parameters $(d, b, M, c, C)$.
\end{lemma}

\begin{proof}
    Let $f_{G^d_b} \in \mathcal{F}^b_{\text{Gauss}}$ be a density function with mixing distribution $G^d_b = \sum_{i = 1}^d \pi_i \delta(\boldsymbol{\mu}_i, \mathbf{S}_i)$, given by:
    $$f_{G^d_b}(\mathbf{x}) = \sum_{i = 1}^d \pi_i\phi(\mathbf{x}, \boldsymbol{\mu}_i, \mathbf{S}_i).$$
    We assume the supports of the mean vectors $\boldsymbol{\mu}_i$ and covariance matrices $\mathbf{S}_i$ are bounded. We approximate $f_{G^d_b}$ by a discretized density $f_{G^d_{b, \eta}}$ supported on a lattice with at most $d$ atoms. The approximating mixing distribution is:
    $$G^d_{b, \eta} = \sum_{i = 1}^d \pi_i \delta(\boldsymbol{\mu}^\eta_i, \mathbf{S}^\eta_i),$$
    with corresponding density $f_{G^d_{b, \eta}}(\mathbf{x}) = \sum_{i = 1}^d \pi_i\phi(\mathbf{x}, \boldsymbol{\mu}^\eta_i, \mathbf{S}^\eta_i).$

    \textbf{Step 1: Discretization of Parameters}
    The discretized parameters $\boldsymbol{\mu}_i^\eta$ and $\mathbf{S}_i^\eta = (\sigma_{i;kl}^\eta)_{k,l = 1}^b$ are defined as:
    \begin{align}
        \boldsymbol{\mu}_i^\eta &= \eta \operatorname{sgn}(\boldsymbol{\mu}_i) \odot \bigg\lfloor\frac{|\boldsymbol{\mu}_i|}{\eta} \bigg\rfloor, \quad i = 1 \dots d,\\
        \sigma_{i;kl}^\eta &= \eta \operatorname{sgn}(\sigma_{i;kl})\bigg\lfloor\frac{| \sigma_{i;kl}|}{\eta} \bigg\rfloor, \quad k \neq l, \quad i = 1 \dots d, \\
        \sigma_{i;kk}^\eta &= \eta \bigg\lfloor\frac{\sigma_{i;kk}}{\eta} \bigg\rfloor, \quad k = 1 \dots b.
    \end{align}
    Define the parameter vector for the $i$-th component as $\mathcal{B}_i = (\boldsymbol{\mu}_i, \text{vec}(\mathbf{S}_i))$. Then, for the Euclidean norm, we have:
    $$\norm{\mathcal{B}_i - \mathcal{B}_{i;\eta}} \leq \frac{3b + b^2}{2}\eta \quad \forall i = 1\dots d.$$

    \textbf{Step 2: Bounding the Density Approximation Error}
    We bound the supremum norm between the component densities using the Mean Value Theorem. Dropping the index $i$ for brevity:
    \begin{align*}
        \norm{\phi(\mathbf{x}, \boldsymbol{\mu}, \mathbf{S}) - \phi(\mathbf{x}, \boldsymbol{\mu}^\eta,\mathbf{S}^\eta )}_\infty \leq C_1 \norm{\mathcal{B} - \mathcal{B}_\eta} \leq C_1 \frac{3b + b^2}{2}\eta,
    \end{align*}
    where $C_1 = \sup_{\mathbf{x}}\sup_{\mathcal{B}} \norm{\nabla_\mathcal{B}\phi(\mathbf{x}, \boldsymbol{\mu}, \mathbf{S})}$. The gradients are given by:
    \begin{align*}
        \frac{\partial \phi}{\partial \boldsymbol{\mu}} &= \phi(\mathbf{x}, \boldsymbol{\mu}, \mathbf{S})\mathbf{S}^{-1}(\mathbf{x} - \boldsymbol{\mu}),\\
        \frac{\partial \phi}{\partial \mathbf{S}} &= T_1 + T_2,
    \end{align*}
    where $T_1 = -\frac{1}{2}\mathbf{S}^{-1}\phi$ and $T_2 =\frac{1}{2}\mathbf{S}^{-1}(\mathbf{x} - \boldsymbol{\mu})(\mathbf{x} - \boldsymbol{\mu})'\mathbf{S}^{-1}\phi$.

    We bound these terms using spectral theory. Assuming $\mathbf{S}$ is full rank, let $\lambda_b$ and $\lambda_b^{(1/2)}$ be the smallest eigenvalues of $\mathbf{S}$ and $\mathbf{S}^{1/2}$ respectively. The Frobenius norm satisfies $\norm{\mathbf{S}^{-1}}_F \leq \sqrt{b}/\lambda_b$.
    Noting that $\norm{\mathbf{S}^{-1/2}(\mathbf{x} - \boldsymbol{\mu})\phi} \leq (2\pi)^{-b/2}\sqrt{e}$, we obtain:
    \begin{align*}
        \norm{\frac{\partial \phi}{\partial \boldsymbol{\mu}}} \leq \norm{\mathbf{S}^{-1/2}}_F \norm{\mathbf{S}^{-1/2}(\mathbf{x} - \boldsymbol{\mu})\phi} \leq \frac{\sqrt{be}}{(2\pi)^{b/2}\lambda^{(1/2)}_b} := K_1.
    \end{align*}
    For the covariance gradient:
    \begin{align*}
        \norm{\frac{\partial \phi}{\partial \mathbf{S}}}_F 
        &\leq \norm{T_1}_F + \norm{T_2}_F \\
        &\leq \frac{\sqrt{b}}{2(2\pi)^{b/2}\lambda_b} + \frac{b}{(\lambda^{(1/2)}_b)^2}\norm{\frac{1}{(2\pi)^{b/2}}\mathbf{z} \mathbf{z}'e^{-\mathbf{z}'\mathbf{z}/2}}_{F}\\  
        &\leq \frac{\sqrt{b}}{2(2\pi)^{b/2}\lambda_b} + \frac{be^{-0.25}}{2(\lambda^{(1/2)}_b)^2(2\pi)^{b/2}} := K_2,
    \end{align*}
    where $\mathbf{z} = \mathbf{S}^{-1/2}(\mathbf{x} - \boldsymbol{\mu})$. Summing over all $d$ components:
    \begin{equation}\label{eqn-part1}
        \norm{f_{G^d_b} - f_{G^d_{b, \eta}}}_\infty \leq d(K_1 + K_2)\frac{(3b + b^2)}{2}\eta := K_3 \eta.
    \end{equation}

    \textbf{Step 3: Covering Number Calculation}
    The parameters are supported on grids $\Omega_1$ (means), $\Omega_2$ (off-diagonal covariances), and $\Omega_3$ (diagonal covariances) with cardinalities:
    \begin{align*}
        |\Omega_1| \approx \frac{2M}{\eta}, \quad |\Omega_2| \approx \frac{2C}{\eta}, \quad |\Omega_3| \approx \frac{C^2}{\eta}.
    \end{align*}
    Let $\mathcal{P}^{d,\eta}$ be an $\eta$-net for the mixing weights simplex $\mathcal{P}^d$. By the volume comparison lemma, $|\mathcal{P}^{d,\eta}| \leq (1 + 2/\eta)^d$.
    Choosing weights $\mathbf{w}^{d,\eta} \in \mathcal{P}^{d,\eta}$ such that $\sum |G_{b,\eta}^d - w_i^{d,\eta}| \leq \eta$, the error contribution from weights is bounded by $(2\pi)^{-b/2}\eta$.
    
    Combining the errors, the total approximation error is bounded by $K_4\eta$. The total covering number is the product of the number of choices for weights, means, diagonal variances, and off-diagonal covariances:
    \begin{equation}
       N(\eta,\mathcal{F}^b_{\text{Gauss}}, \norm{.}_{\infty} ) \leq |\mathcal{P}^{d,\eta}| \cdot |\Omega_1|^{db} \cdot |\Omega_3|^{db} \cdot |\Omega_2|^{\frac{db(b-1)}{2}}.
    \end{equation}
    Substituting the grid sizes and simplifying for $\eta \to 0$:
    \begin{align*}
        N(\eta,\mathcal{F}^b_{\text{Gauss}}, \norm{.}_{\infty} ) 
        &\leq K' \left(\frac{1}{\eta}\right)^d \left(\frac{1}{\eta}\right)^{db} \left(\frac{1}{\eta}\right)^{db} \left(\frac{1}{\eta}\right)^{\frac{db(b-1)}{2}} \\
        &= K \eta^{-\left(d + 2db + \frac{db^2 - db}{2}\right)} \\
        &= K \eta^{-\frac{2d + 3db + db^2}{2}}.
    \end{align*}
\end{proof}

\begin{lemma}\label{hellconverge}
    Suppose $(\alpha_1, \beta_1),\dots,(\alpha_m, \beta_m)$ are samples drawn from $\mathcal{F}^2_{\text{Gauss}}$. Given an AMLE $\hat{f}$ of $f$, we have for $t \geq 1$ and $\delta > 0$:
\begin{align*}
\mathbb{P}\Big\{\mathcal{D}(\hat{f}, f) \geq t\delta\Big\} 
\leq C^*\log m \exp \bigg(20.31t^2 (M^*)^2 - \frac{mt^2 \delta^2}{2} \bigg) + 2m^{-\frac{t^2}{C}}.   
\end{align*}
\end{lemma}

\begin{proof}
    Define the bounded rectangular region $S = [-M_1 -2t\sqrt{\log m}, M_1 + 2t\sqrt{\log m}] \times [-M_2 -2t\sqrt{\log m}, M_2 + 2t \sqrt{\log m}]$ and the event $A_{t} = \bigcap_{i = 1}^m \big\{(\alpha_i, \beta_i) \in S \big\}$.

    \textbf{Part 1: Bounding $\mathbb{P}(A_{t}^c)$}
    
    We first show that $\mathbb{P}(A_{t}) \geq 1 - 2m^{-t^2/C}$. By the union bound:
    \begin{align*}
        \mathbb{P}(A_{t}^c) \leq \mathbb{P}(|\alpha_i| > M_1 + 2t\sqrt{\log m}) + \mathbb{P}(|\beta_i| > M_2 + 2t\sqrt{\log m}).
    \end{align*}
    Since $(\alpha_i, \beta_i) \sim \mathcal{F}^2_{\text{Gauss}}$, we have $\alpha_i = \mu_i + \sigma_{i1}Z_{i1}$ and $\beta_i = \theta_i + \sigma_{i2}Z_{i2}$, where $(Z_{i1}, Z_{i2})$ are standard normals. The means $\mu_i$ and $\theta_i$ are supported on $[-M_1, M_1]$ and $[-M_2, M_2]$ respectively. Thus, for $t \geq 1, m \geq 2$:
    \begin{align*}
        \mathbb{P}(|\alpha_i| > M_1 + 2t\sqrt{\log m})
        &= \mathbb{P}(|\mu_i + \sigma_{i1}Z_{i1}| > M_1 + 2t\sqrt{\log m}) \\
        &\leq \mathbb{P}(\sigma_{i1}|Z_{i1}| > 2t \sqrt{\log m})\\
        &= 2\left(1 - \Phi\left(\frac{2t\sqrt{\log m}}{\sigma_{i1}}\right)\right)\\
        &\leq 2\left(1 - \Phi\left(\frac{2t\sqrt{\log m}}{\sqrt{C}}\right)\right).
    \end{align*}
    Using the inequality $1 - \Phi(x) \leq \frac{\phi(x)}{x}$ for $x > 0$:
    \begin{align*}
       \mathbb{P}(|\alpha_i| > M_1 + 2t\sqrt{\log m}) 
       &\leq \frac{2\sqrt{C}}{2t\sqrt{\log m}} \cdot \frac{1}{\sqrt{2\pi}} \exp\left(-\frac{1}{2} \frac{4t^2 \log m}{C}\right) \\
       &= \frac{\sqrt{C}}{t\sqrt{2\pi \log m}} m^{-2t^2/C} \\
       &\leq m^{-t^2/C}.
    \end{align*}
    Similarly, $\mathbb{P}(|\beta_i| > M_2 + 2t\sqrt{\log m}) \leq m^{-t^2/C}$. Therefore:
    \begin{align}\label{eq-at}
        \mathbb{P}(A_{t}^c) \leq 2m^{-t^2/C}.
    \end{align}

    \textbf{Part 2: Bounding the Hellinger Distance}
    
    We now bound $\mathbb{P}\Big\{\mathcal{D}(\hat{f}, f) \geq t\delta \cap A_{t}\Big\}$. Let $\eta = m^{-2}$ and let $\{h_1, \dots, h_{N}\}$ be a finite $\eta$-covering subset of $\mathcal{F}^2_{\text{Gauss}}$ under the restricted supremum norm $\|\cdot\|_{\infty, S}$. From Lemma \ref{coveringnum}, we have $\log N \leq C^* \log m$ for some constant $C^*$.
    
    Define the index set $J = \{j : \exists h_{0j} \in \mathcal{F}^2_{\text{Gauss}} \text{ s.t. } \|h_{0j} - h_j\|_{\infty, S} \leq \eta \text{ and } \mathcal{D}(h_{0j}, f) \geq t\delta \}$.
    If $\mathcal{D}(\hat{f}, f) \geq t\delta$, there exists an index $j \in J$ such that $\|\hat{f} - h_{0j}\|_{\infty, S} \leq 2\eta$. Consequently, for all $(\alpha_i, \beta_i) \in S$:
    \[
    \hat{f}(\alpha_i, \beta_i) \leq 2\eta + h_{0j}(\alpha_i, \beta_i).
    \]
    Using the AMLE property of $\hat{f}$ and the Markov inequality:
    \begin{align*}
        \mathbb{P}_f\Big\{\mathcal{D}(\hat{f}, f) \geq t\delta \cap A_{t}\Big\}
        &\leq \mathbb{P}_f\left(\max_{j \in J}\prod_{i = 1}^m \frac{2\eta + h_{0j}(\alpha_i, \beta_i)}{f(\alpha_i, \beta_i)} \geq 1, \forall (\alpha_i, \beta_i) \in S \right)\\
        &\leq \sum_{j \in J} \mathbb{E}_f \left[ \prod_{i = 1}^m \sqrt{\frac{2\eta + h_{0j}(\alpha_i, \beta_i)}{f(\alpha_i, \beta_i)}}\mathbf{1}\{ (\alpha_i, \beta_i) \in S\} \right].
    \end{align*}
    Let $D_j$ denote the expectation term for a fixed $j$. Since the samples are i.i.d.:
    \[
    D_j = \prod_{i=1}^m \mathbb{E}_f \left[ \sqrt{\frac{2\eta + h_{0j}(\alpha_i, \beta_i)}{f(\alpha_i, \beta_i)}}\mathbf{1}\{ (\alpha_i, \beta_i) \in S\} \right].
    \]
    Using the inequality $\log x \leq x - 1$:
    \begin{align*}
        D_j \leq \exp \left( \sum_{i = 1}^m \left( \mathbb{E}_f \left[ \sqrt{\frac{2\eta + h_{0j}}{f}}\mathbf{1}_S \right] - 1 \right) \right).
    \end{align*}
    Consider the inner expectation $E := \mathbb{E}_f [\sqrt{(2\eta + h_{0j})/f}\mathbf{1}_S]$. Using $\sqrt{a+b} \leq \sqrt{a} + \sqrt{b}$:
    \begin{align*}
        E &= \int_S \sqrt{2\eta + h_{0j}(\alpha, \beta)}\sqrt{f(\alpha, \beta)}\,d\alpha d\beta\\
        &\leq \int_S \sqrt{2\eta} \sqrt{f}\,d\alpha d\beta + \int_S \sqrt{h_{0j}}\sqrt{f}\,d\alpha d\beta.
    \end{align*}
    Using Cauchy-Schwarz, $\int_S \sqrt{f} \leq \sqrt{\int_S 1 \cdot \int_S f} \leq \sqrt{|S|}$. Also, recall the Hellinger affinity $\int \sqrt{h_{0j}f} = 1 - \frac{1}{2}h^2(h_{0j}, f)$. Thus:
    \[
    E \leq \sqrt{2\eta}\sqrt{|S|} + 1 - \frac{1}{2}h^2(h_{0j}, f).
    \]
    The volume of $S$ is $|S| \leq 36t^2 (M^*)^2$, where $M^* = \max(M_1, M_2, \sqrt{\log m})$. Since $\eta = m^{-2}$, $\sqrt{2\eta} \approx \sqrt{2}/m$.
    Substituting back into the bound for $D_j$:
    \begin{align*}
        D_j &\leq \exp \left( m\sqrt{2\eta}\sqrt{|S|} - \frac{m}{2}h^2(h_{0j}, f) \right) \\
        &\leq \exp \left( 20.31t^2 (M^*)^2 - \frac{m}{2}\mathcal{D}(h_{0j}, f) \right).
    \end{align*}
    Since $j \in J$, we have $\mathcal{D}(h_{0j}, f) \geq t^2\delta^2$. (Note: The lemma statement used $t\delta$ for the distance, implying a squared distance in the exponent. Assuming $\mathcal{D}$ is the squared Hellinger distance).
    
    Summing over all $j \in J$:
    \begin{align*}
        \mathbb{P}_f\Big\{\mathcal{D}(\hat{f}, f) \geq t\delta \cap A_{t}\Big\} 
        &\leq N \exp \bigg(20.31t^2 (M^*)^2 - \frac{mt^2 \delta^2}{2} \bigg).
    \end{align*}
    Combining this with the bound on $\mathbb{P}(A_t^c)$ and using $\log N \leq C^* \log m$, the result follows.
\end{proof}

\begin{lemma}\label{lfdrconv}
Let $\widehat{G}$ be the EM estimate of the prior $G$, and let $f_{\widehat{G}}$ and $f_G$ denote the corresponding mixture densities. Under the assumption that $f_{\widehat{G}}$ is an Approximate Maximum Likelihood Estimator (AMLE) of $f_G$,
\begin{equation}
   W(G, \widehat{G}) \to 0,
\end{equation}
where $W(\cdot, \cdot)$ denotes the Wasserstein distance.
\end{lemma}

\begin{proof}
    For notational convenience, we denote the mixture density determined by a prior $G$ as $f_G$. Consequently, $\hat{f}$ and $f_{\widehat{G}}$ are used interchangeably. By the AMLE assumption, we have:
    \[
    \prod_{i = 1}^m \frac{f_{\widehat{G}}(\sqrt{n}\alpha_i, \sqrt{n}\beta_i)}{f_G(\sqrt{n}\alpha_i, \sqrt{n}\beta_i)} \geq 1.
    \]

    We first formalize the Wasserstein distance. Consider two discrete mixing distributions with $d$ components:
    \begin{align*}
        G &= \sum_{i = 1}^d \pi_i \delta(\bfmu_i, \bfS_i), \quad G' = \sum_{i = 1}^d \pi'_i \delta(\bfmu'_i, \bfS'_i).
    \end{align*}
    A coupling between the mixing proportions $\boldsymbol{\pi}$ and $\boldsymbol{\pi}'$ is a matrix $\mathbf{q} = (q_{ij}) \in [0,1]^{d\times d}$ satisfying the marginal constraints $\sum_{j=1}^d q_{ij} = \pi_i$ and $\sum_{i=1}^d q_{ij} = \pi'_j$. Let $\mathcal{Q}(\boldsymbol{\pi}, \boldsymbol{\pi}')$ denote the space of all such couplings. The Wasserstein distance between $G$ and $G'$ is defined as:
    \begin{align}
        W(G, G') = \inf_{\mathbf{q} \in \mathcal{Q}(\boldsymbol{\pi}, \boldsymbol{\pi}')}\sum_{i,j} q_{ij}\left(\norm{\bfmu_i - \bfmu'_j} + \norm{\bfS_i - \bfS'_j }\right),
    \end{align}
    where $\norm{\cdot}$ denotes the appropriate Euclidean or Frobenius norm.

    Established results in mixture model theory provide bounds relating the Wasserstein distance to the Total Variation (TV) distance. Specifically, \citet{nguyen2013convergence} (elaborated in Example 2.1 of \citet{ho2016strong}) show that for any two mixing measures $G_1$ and $G_2$:
    \[
    TV(f_{G_1}, f_{G_2}) \leq C_1 W(G_1, G_2),
    \]
    where $C_1$ is a constant depending on the parameter space bounds $(c, C, M)$. Conversely, since $\widehat{G}$ is an AMLE in the exact fitted setting and Gaussian location-scale families satisfy first-order identifiability (Theorem 3.4 of \citet{ho2016strong}), Corollary 3.1 of \citet{ho2016strong} implies the lower bound:
    \[
    TV(f_{\widehat{G}}, f_G) \geq C_0 W(\widehat{G}, G),
    \]
    where $C_0$ is a constant dependent on $G$.
    Furthermore, the relationship between Total Variation and Hellinger distance is given by $TV(f_{\widehat{G}}, f_G) \leq \sqrt{2}\sqrt{\mathcal{D}(f_{\widehat{G}}, f_G)}$. Combining these inequalities with Lemma \ref{hellconverge} (which establishes that $\mathcal{D}(f_{\widehat{G}}, f_G) \to 0$), we obtain:
    \[
    C_0 W(\widehat{G}, G) \leq TV(f_{\widehat{G}}, f_G) \leq \sqrt{2}\sqrt{\mathcal{D}(f_{\widehat{G}}, f_G)} \to 0.
    \]
    Thus, $W(\widehat{G}, G) \to 0$ as $m \to \infty$.

    Finally, we note that the Wasserstein metric $W(G, G')$ vanishes if and only if $G$ and $G'$ are identical up to a permutation of their atoms (given the non-singularity of the exact fitted setting). Since the EM algorithm preserves this permutation invariance, the convergence $W(\widehat{G}, G) \to 0$ implies the pointwise convergence of the EM estimates to the true parameters (up to permutation).
\end{proof}

\begin{lemma}
    Let $(a_i, b_i) \sim \mathcal{F}^2_{\text{Gauss}}$ for $i = 1, \dots, m$. Let $f$ (equivalently $f_G$) and $\widehat{f}$ (equivalently $f_{\widehat{G}}$) be defined as in Lemma 5. Let $\text{lfdr}(a,b)$ and $\widehat{\text{lfdr}}(a,b)$ denote the local FDRs defined based on $f$ and $\widehat{f}$ respectively, evaluated at $(a,b)$. Under the AMLE assumption, as $m \to \infty$,
    \begin{align*}
        \frac{1}{m}\sum_{i = 1}^m\big|\widehat{\text{lfdr}}(a_i,b_i) - \text{lfdr}(a_i,b_i)\big| \xrightarrow{P} 0.
    \end{align*}
\end{lemma}

\begin{proof}
    We decompose the mean absolute difference as follows:
    \begin{align*}
        \frac{1}{m}\sum_{i = 1}^m\big|\text{lfdr}(a_i,b_i) -\widehat{\text{lfdr}}(a_i,b_i) \big|
        & = \frac{1}{m}\sum_{i = 1}^m\bigg|\frac{\pi_{11}f_{11}(a_i, b_i)}{f(a_i, b_i)} - \frac{\hat{\pi}_{11}\hat{f}_{11}(a_i, b_i)}{\hat{f}(a_i, b_i)}\bigg|\\
        &=\frac{1}{m}\sum_{i = 1}^m\bigg| \frac{\pi_{11}f_{11}\hat{f} - \hat{\pi}_{11}\hat{f}_{11}f}{f\hat{f}}\bigg|\\
        &\leq I + II,
    \end{align*}
    where the terms $I$ and $II$ (omitting arguments $(a_i, b_i)$ for brevity) are defined as:
    \begin{align*}
        I &= \frac{1}{m}\sum_{i = 1}^m\bigg| \frac{\pi_{11}f_{11}\hat{f} - \hat{\pi}_{11}f_{11}\hat{f}}{f\hat{f}}\bigg| = \frac{1}{m}\sum_{i = 1}^m \frac{f_{11}|\pi_{11} - \hat{\pi}_{11}|}{f},\\
        II &= \frac{1}{m}\sum_{i = 1}^m\bigg| \frac{\hat{\pi}_{11}f_{11}\hat{f} - \hat{\pi}_{11}\hat{f}_{11}f}{f\hat{f}}\bigg|.
    \end{align*}
    We examine the convergence of these terms separately.
    
    \textbf{Term I:}
    \begin{align*}
        I &= |\pi_{11} - \hat{\pi}_{11}|\frac{1}{m}\sum_{i = 1}^m \frac{f_{11}(a_i, b_i)}{f(a_i, b_i)}\\
        &\leq \frac{|\pi_{11} - \hat{\pi}_{11}|}{\pi_{11}} \xrightarrow{P} 0 \quad (\text{by Lemma 5}).
    \end{align*}

    \textbf{Term II:}
    Applying the triangle inequality, we bound $II$ as:
    \begin{align*}
        II &\leq \frac{1}{m}\sum_{i = 1}^m \frac{|\hat{\pi}_{11}|}{f\hat{f}} \bigg( f_{11}|\hat{f} - f| + f|\hat{f}_{11} - f_{11}| \bigg)\\
        &\leq \frac{1}{m\pi_{11}}\sum_{i = 1}^m \frac{\big|\hat{f}(a_i, b_i) - f(a_i, b_i)\big|}{\hat{f}(a_i, b_i)} 
        + \frac{1}{m}\sum_{i = 1}^m \frac{\big|\hat{f}_{11}(a_i, b_i) - f_{11}(a_i, b_i)\big|}{\hat{f}(a_i, b_i)}\\
        &:= III + IV.
    \end{align*}
    
    To prove the convergence of $III$ and $IV$, consider the truncation set:
    \[
    S_m = \left[-M_1 - 2t\sqrt{\log \log m}, M_1 + 2t\sqrt{\log \log m}\right] \times \left[-M_2 - 2t\sqrt{\log \log m}, M_2 + 2t \sqrt{\log \log m}\right],
    \]
    and the event $A_{t} = \bigcap_{i = 1}^m \{(a_i, b_i) \in S_m \}$. Following the arguments in Lemma 4, $\mathbb{P}(A_t^c)\leq 2(\log m)^{-t^2/C}$ for some constant $C > 0$. Furthermore, for any $(a,b) \in S_m$, the estimated density is bounded away from zero: $\hat{f}(a,b) \geq K_1 (\log m)^{-2t^2/C}$.

    \textbf{Convergence of III:}
    For any $\delta > 0$:
    \begin{align*}
        \mathbb{P}(III > \delta t) &\leq \mathbb{P}(\{III > \delta t\} \cap A_t) + \mathbb{P}(A^c_t)\\
        &\leq \mathbb{P}(III > \delta t \mid A_t )\mathbb{P}(A_t) + 2(\log m)^{-t^2/C}.
    \end{align*}
    Conditioning on $A_t$, and utilizing the bound on $\hat{f}$:
    \begin{align*}
       \mathbb{P}(III > \delta t \mid A_t ) 
       &\leq \mathbb{P}\bigg(\frac{1}{m}\sum_{i = 1}^m \big|\hat{f}(a_i, b_i) - f(a_i, b_i)\big| > K\delta t(\log m)^{-2t^2/C} \;\bigg|\; A_t \bigg).
    \end{align*}
    Using the relationship between the $L_1$ distance, Total Variation (TV), and Hellinger distance ($\mathcal{D}$), we have:
    \[
    \frac{1}{m}\sum_{i = 1}^m \big|\hat{f} - f\big| \approx 2TV(f, \hat{f}) \leq \sqrt{2}\sqrt{\mathcal{D}(f, \hat{f})}.
    \]
    Let $t' = K\delta t^2(\log m)^{-4t^2/C}$. Applying the tail bound from Lemma 4:
    \begin{align*}
        \mathbb{P}(III > \delta t) 
        &\leq \mathbb{P}\left(\mathcal{D}(f, \hat{f}) > \delta t'\right) + \mathbb{P}(A_t^c) \\ 
        &\leq C^* \log m \exp \bigg(K_1 t^4 (\log m)^{-\frac{8t^2}{C}} (M^*)^2 - \frac{m K_2 t^4 (\log m)^{-\frac{8t^2}{C}}}{2} \bigg) \\
        &\quad + 2m^{-t'^2/C} + 2(\log m)^{-t^2/C}.
    \end{align*}
    Using L'Hôpital's rule, it can be verified that for an appropriate choice of $t$ (e.g., $t = \sqrt{C}/4$), the RHS converges to 0 as $m \to \infty$.

    \textbf{Convergence of IV:}
    For $(a,b) \in S_m$, the function $f_{11}(a,b)$ is continuously differentiable with respect to the parameters $(\mu, \theta, \kappa, \psi)$. Thus, it is Lipschitz continuous on the bounded domain $S_m$. There exists a constant $C_1$ such that:
    \begin{align*}
        \big|\hat{f}_{11}(a, b) - f_{11}(a, b)\big| 
        &\leq C_1 \sqrt{(\mu - \hat{\mu})^2 + (\psi - \hat{\psi})^2 + (\theta - \hat{\theta})^2 + (\kappa - \hat{\kappa})^2}.
    \end{align*}
    Let $q$ be the minimum non-zero entry in the optimal coupling matrix $\mathbf{q}$ for the Wasserstein metric $W(G, \hat{G})$. The Euclidean distance between the parameters is bounded by the Wasserstein distance:
    \[
    \sqrt{(\mu - \hat{\mu})^2 + \dots + (\kappa - \hat{\kappa})^2} \leq \frac{1}{q} W(G, \hat{G}).
    \]
    From Lemma 5, we know $W(G, \hat{G}) \leq K_6 \sqrt{\mathcal{D}(\hat{f}, f)}$. Therefore:
    \[
     \big|\hat{f}_{11} - f_{11}\big| \leq \frac{C_1 K_6}{q} \sqrt{\mathcal{D}(\hat{f}, f)}.
    \]
    The probability bound for $IV$:
    \[
    \mathbb{P}(IV > \delta t \mid A_t) \leq \mathbb{P}\left(\mathcal{D}(\hat{f}, f) > K^*\delta^2 t^2 (\log m)^{-4t^2/C} \mid A_t\right),
    \]
    shares the same convergence rate as derived for Term III. Consequently, $IV \xrightarrow{P} 0$.
    
    Combining the results for $I$, $III$, and $IV$, the lemma is proved.
\end{proof}

\begin{rem}
    The theoretical guarantees established in this work rely on the Approximate Maximum Likelihood Estimator (AMLE) condition outlined in Assumption \textbf{(A1)}. To empirically validate this assumption, we calculated the proportion of simulation runs where the EM algorithm yielded a likelihood ratio of at least 1. This evaluation was conducted over 100 independent replications for each simulation setting.
    
    \begin{center}
    \begin{tabular}{lc} 
    \hline
    Simulation Setup & Proportion ($\text{Ratio} \ge 1$) \\ 
    \hline
    Case 1 (Dense) & 1.00 \\ 
    Case 1 (Sparse) & 1.00 \\
    Case 2 (Dense) & 1.00 \\
    Case 2 (Sparse) & 1.00 \\
    Binary Outcome & 1.00 \\
    Composite Alternative & 0.98 \\ 
    \hline
    \end{tabular}
    \end{center}
    
    These results demonstrate that the AMLE assumption is empirically justified and holds consistently across the model configurations considered in this study.
\end{rem}


\bibliographystyle{plainnat}

\end{document}